\documentclass[11pt]{article}

\usepackage{amssymb,amsmath,amsthm}
\usepackage{mathrsfs}
\usepackage[colorlinks,allcolors=blue]{hyperref}
\usepackage{graphicx}
\usepackage{amsfonts,amscd}
\usepackage[all,cmtip]{xy}
\usepackage{cleveref}
\usepackage{xcolor}
\usepackage{arydshln}
\usepackage{comment}
\usepackage{cases}
\definecolor{mg}{rgb}  {0.85, 0.,  0.85}

\def\a{\mathfrak a}
\def\AA{\mathcal A}
\def\B{\mathscr B}
\def\C{\mathbb C}
\def\CC{\mathcal C}
\def\d{\mathrm d}

\def\Dt{\mathfrak D}
\def\EE{\mathcal E}
\def\f{\mathfrak f}
\def\F{\mathscr F}
\def\G{\mathcal G}
\def\H{\mathcal H}
\def\Hrond{\mathscr H}
\def\h{\mathfrak h}
\def\K{\mathscr K}
\def\Kt{\mathfrak K}
\def\KK{\mathcal K}

\def\N{\mathbb N}
\def\NN{\mathrm N}

\def\P{\mathcal P}
\def\QQ{\mathcal Q}
\def\R{\mathbb R}
\def\S{\mathbb S}
\def\s{\mathfrak s}
\def\T{\mathcal T}

\def\u{\mathfrak u}

\def\v{\mathfrak v}
\def\V{\mathcal V}
\def\Xt{\mathfrak X}
\def\Z{\mathbb Z}
\newcommand{\ox}{\otimes}

\def\Tr{\mathrm{Tr}}

\def\diag{\mathop{\mathrm{diag}}\nolimits}
\def\e{\mathop{\mathrm{e}}\nolimits}

\def\Ran{\mathop{\mathrm{Ran}}\nolimits}

\DeclareMathOperator*{\slim}{s\hspace{0.1pt}-\hspace{0.1pt}lim}
\def\supp{\mathop{\mathrm{supp}}\nolimits}

\def\ltwo{\mathop{\mathrm{L}^2}\nolimits}

\def\arctanh{\mathop{\mathrm{arctanh}}\nolimits}

\def\sp{\mathop{\mathrm{span}}\nolimits}
\def\tl{\tilde \lambda^\theta}
\def\tP{\tilde{\mathcal{P}}^\theta}

\def\tH{{{\tilde \Hrond}^\theta}}

\def\Arg{\mathrm{Arg}}
\def\Var{\mathrm{Var}}
\def\ind{\mathop{\mathrm{ind}}\nolimits}

\newtheorem{Theorem}{Theorem}[section]
\newtheorem{Remark}[Theorem]{Remark}
\newtheorem{Lemma}[Theorem]{Lemma}
\newtheorem{Corollary}[Theorem]{Corollary}
\newtheorem{Proposition}[Theorem]{Proposition}
\newtheorem{Definition}[Theorem]{Definition}

\begin{document}


\title{Topological Levinson's theorem in presence of embedded thresholds
and discontinuities of the scattering matrix}

\author{V. Austen${}^1$, D. Parra${}^2$, A. Rennie${}^3$\footnote{Supported by ARC Discovery grant DP220101196.}, S. Richard${}^4$\footnote{Supported by JSPS Grant-in-Aid for scientific research C no 21K03292.}}

\date{\small}
\maketitle
\vspace{-1cm}

\begin{quote}
\begin{itemize}
\item[1] Department of chemistry, Graduate School of Science, Nagoya University, Furo-cho, Chikusa-ku, Nagoya, 464-8602, Japan
\item[2] 
Departamento de Matem\'atica y Estad\'istica, Universidad de La Frontera, \\ Av. Francisco Salazar 01145, Temuco, Chile
\item[3] School of Mathematics and Applied Statistics, University of Wollongong,\\
	Wollongong, Australia
\item[4] Institute for Liberal Arts and Sciences \& Graduate School of Mathematics, Nagoya University, Furo-cho, Chikusa-ku, Nagoya, 464-8601, Japan
\item[] E-mail:  austen.vic.a9@s.thers.ac.jp, daniel.parra.v@usach.cl, renniea@uow.edu.au, richard@math.nagoya-u.ac.jp
\end{itemize}
\end{quote}

\maketitle


\begin{abstract}
A family of discrete Schr\"odinger operators is investigated through scattering theory.
The continuous spectrum of these operators exhibit changes of multiplicity, and some
of these operators possess resonances at thresholds.
It is shown that the corresponding wave operators belong to an
explicitly constructed $C^*$-algebra, whose K-theory is carefully analysed.
An index theorem is deduced from these investigations, which corresponds
to a topological version of Levinson's theorem in presence of embedded thresholds, resonances, and  changes of multiplicity of the scattering matrices.
In the second half of the paper, very detailed computations for the simplest realisation of this
family of operators are provided. In particular, a surface of resonances is  
exhibited, probably for the first time.
For Levinson's theorem, it is shown that contributions due to resonances at the lowest value and 
at the highest value of the continuous spectrum play an essential role.
\end{abstract}


\smallskip


\tableofcontents

\section{Introduction}\label{sec:intro}
\setcounter{equation}{0}

In this paper we prove a version of Levinson's theorem for scattering systems with 
non-constant multiplicity of the continuous spectrum. These systems appeared first in \cite{NRT}.
Levinson’s theorem relates the number of bound states of a quantum
mechanical system to the scattering part of the system.
The original formulation was by N. Levinson in \cite{Lev} in the context of
a Schr\"odinger operator with a spherically symmetric potential.
In \cite{KR06, KR07, KR08}, a topological interpretation of Levinson's theorem was
proposed. The topological approach consists in constructing a $C^*$-algebraic framework for the scattering system involving not 
only the scattering operator, but also new operators that describe the system at thresholds energies. The number of bound states is the Fredholm 
index of one wave operator, evaluated via a winding number.

The review paper \cite{Ri} describes the topological approach and applications to several classical models, see also \cite{AR23,BSB,IT19,SB16} for similar investigations. 
A common feature of the examples considered in \cite{Ri}, and all previous examples,
is that the continuous spectrum of the underlying self-adjoint operators is made of one connected set,  with no change of multiplicity.
The topological approach has also been applied to more exotic situations, for example to Dirac
operators (with an underlying spectrum made of two disjoint connected parts) \cite{PR}, 
to systems with embedded or complex eigenvalues \cite{NPR,RT10}, or to systems with an infinite number of eigenvalues \cite{IR1,IR2}. 

The non-constancy of the multiplicity of continuous spectrum incorporated here has various drastic effects: the spectral representation of the 
underlying unperturbed operator takes place in a direct integral with fibres
of non-constant dimension; the scattering matrices are acting on different Hilbert spaces with
jumps in their dimension; the  $C^*$-algebras containing the wave operators and  the scattering operator  have
a non-trivial internal structure. The notion of continuity of the scattering operator, 
which has played a crucial role in all investigations on Levinson's theorem, has
also to be suitably adapted.

These interesting new features arise in
a family of systems that were introduced and partially studied in \cite{NRT} as parts of a discrete model 
of scattering theory in a half-space, see also \cite{Cha00,CS00,Fra03,Fra04,JL00,RT16,RT17}. 
Here we consider each system individually, meaning
that we have at hand a family of scattering systems indexed by a parameter $\theta \in (0,\pi)$.
Each system carries changes of multiplicity in the continuous spectrum, 
and accordingly embedded thresholds. The resulting scattering matrices
act on some finite dimensional Hilbert spaces
with their dimension depending on the energy parameter. Even if a global notion of continuity
can not be defined in this setting, it has been shown in \cite[Sec.~3.3]{NRT} that a local notion of continuity holds, together with the existence of left or right limits at the points corresponding
to a change of multiplicity (or equivalently to a change of dimension).

The $C^*$-algebra $\EE^\theta$ containing the wave operator is built from functions of energy (with values in matrices of varying dimension) as well as functions of the appropriate  derivative on the energy spectrum. The bound states are computed as the winding number of the image of the wave operator in the quotient of $\EE^\theta$ by the compact operators. The quotient algebra consists of continuous matrix-valued functions on non-trivially overlapping loops (topologically speaking), where the overlaps happen along the energy spectrum of the system. The additional non-overlapping parts of the loops appear at opening\;\!/\;\!closing of channels, and restore continuity so that the winding number can be meaningfully discussed.
The proof that the quotient algebra can be viewed in this way is the first appearance of  ``twisted commutators'' \cite{CoM} in scattering theory.
As a result, we  obtain a numerical equality between the number of eigenvalues of the interacting Hamiltonian and  the winding number of the operator made of the scattering operator together with the opening\;\!/\;\!closing operators. Resonances at thresholds (embedded or not) and possibly embedded eigenvalues are automatically included in our computation. 
At the $C^*$-algebraic level, note that the direct integrals of $C^*$-algebras arising here can be studied using the formalism of pullback algebras, see for example \cite[Ex.~4.10.22]{HR}, allowing the $K$-theory to be computed.

Let us now be more precise about the content of this paper.
In Section \ref{sec:themodel} we introduce the model we shall consider
and recall the main properties exhibited in \cite{NRT}. Roughly speaking, the unperturbed 
system $H_0^\theta$ is a discrete magnetic Schr\"odinger operator acting
on a cylinder of the form $\N \times (\Z/N\Z)$ for any integer $N\geq 2$, with a constant magnetic
field in the $\N$ direction encoded by $\theta\in (0,\pi)$. The perturbed system $H^\theta$ is  a finite rank perturbation with perturbation supported on $\{0\}\times (\Z/N\Z)$. 
The full spectral and scattering theory for the pair of operators $(H^\theta,H^\theta_0)$
has been investigated in \cite{NRT}, and in particular an explicit formula
for the wave operator is recalled in Theorem \ref{thm_formula_theta}. 
Precise formulas for the spectral representation of $H_0^\theta$ are also presented, and  the changes of multiplicity  already appear in the spectral
decomposition of $H_0^\theta$, and hence persist throughout the subsequent analysis. 

In Section \ref{sec:C*} we  define two $C^*$-algebras: the $C^*$-algebra
$\AA^\theta$ which contains (after a unitary transformation) the scattering
operator $S^\theta$; and the $C^*$-algebra $\EE^\theta$ which contains (after a unitary transformation) the wave operator $W_-^\theta$. These algebras are constructed in the spectral representation of $H_0^\theta$, which explains their dependence on $\theta \in (0,\pi)$.
This section also contains a first description of the quotient algebra $\QQ^\theta$ of $\EE^\theta$ 
by the ideal of compact operators.  In particular, the proof that $\QQ^\theta$ consists of matrix-valued functions on a compact Hausdorff space uses techniques to analyse commutators that have not appeared before in the scattering literature.

The quotient  algebra $\QQ^\theta$ can be thought of as \emph{an upside down comb} with $2N$ teeth: the first $N$ of them support the operators corresponding to the opening of new channels of scattering, while the last $N$ teeth support the operators closing channels of scattering. 
The structure of the quotient algebra allows
us to analyse the various possible behaviours of the system at embedded thresholds. 
The $K$-theory for the algebra $\QQ^\theta$ is investigated in Section \ref{sec:Kth}. The algebra is efficiently described by assembling matrix-valued functions on intervals by gluing them at boundaries. This allows the Mayer-Vietoris theorem to be used to compute the $K$-theory inductively.

Section \ref{sec:Lev} contains the newly developed topological Levinson's theorem, 
in presence of embedded thresholds and discontinuities of the scattering matrix.
The number of bound states $\# \sigma_{\rm p}(H^\theta)$ is given by
\[
\# \sigma_{\rm p}(H^\theta)=N-\frac{\#\big\{j\mid \s_{jj}(\lambda^\theta_j\pm 2)=1\big\}}2+ \Var \big(\lambda \to \det S^\theta(\lambda)\big),
\] 
where $\s_{jj}(\lambda^\theta_j\pm 2)$ corresponds to a distinguished entry of the scattering matrix at the threshold energy $\lambda^\theta_j\pm 2$, and where $\Var \big(\lambda \mapsto \det S^\theta(\lambda)\big)$ 
is  the total variation of the argument of the piecewise continuous function $\lambda\mapsto\det S^\theta(\lambda)$. 
We refer to Theorem \ref{thm:Adam} for the details and for more explanations.
The proof is based on the usual winding number argument.
Again, its striking feature is its topological nature, despite the non-continuity of the 
scattering matrix, and even the change of dimension of these matrices.
In this section, we also describe the behaviour of the scattering matrix 
at the opening or at the closing of a new channel of scattering. 

Finally, Section \ref{sec:N=2} illustrates
our findings for $N=2$. In this case, all computations can be explicitly performed. 
The importance of the operators related to the opening or the closing of channels is highlighted. The section also contains an illustration of a surface of resonances.
For each set of parameters (which define the finite rank perturbation) in an open set, there
exists a unique $\theta\in (0,\pi)$ for which the scattering system exhibit a resonance at 
the lowest value of its continuous spectrum. By collecting these exceptional values of $\theta$ one gets the surface 
of resonances presented in Figure \ref{fig:3D}. 
For a fixed set of parameters, we also provide an illustration of the non-trivial dependence
on $\theta$ for the number of bound states of $H^\theta$, see Figure \ref{fig_Vic}.
Note that a reader interested only in topological results can skip this rather long section, 
but we hope that a reader interested in explicitly solvable models will enjoy it.

As a conclusion, this work illustrates the flexibility of the algebraic approach in proving index theorems in scattering theory. The success in dealing with
change of multiplicity opens the door for various new applications, such as  the 
$N$-body problem or highly anisotropic systems. 

\vspace{5mm}
{\bf Acknowledgements} 
\newline
A.~R.~acknowledges the support of the ARC Discovery grant DP220101196.

\section{The model and the existing results}\label{sec:themodel}
\setcounter{equation}{0}

In this section, we introduce the quantum system and give a rather complete description of the results obtained in \cite{NRT}, to which we refer for the details. The following exposition is partially based on the short review paper \cite{RiRIMS}.

We set $\N:=\{0,1,2,\dots\}$. In the Hilbert space $\H:=\ell^2(\N)$
we consider  the discrete Neumann
adjacency operator whose action on $\phi\in\ell^2(\N)$ is described by
$$
\big(\Delta_\NN\;\!\phi\big)(n)=
\begin{cases}
2^{1/2}\;\!\phi(1) &\hbox{if $n=0$}\\
2^{1/2}\;\!\phi(0)+\phi(2) &\hbox{if $n=1$}\\
\phi(n+1)+\phi(n-1) &\hbox{if $n\ge2$.}
\end{cases}
$$
For any fixed $N\in \N$ with $N\geq 2$ and for any fixed $\theta \in (0,\pi)$ 
we also consider the $N\times N$ Hermitian matrix
\begin{equation*}
A^\theta:=
\begin{pmatrix}
0 & 1 & 0 &\cdots & 0 &\e^{-i\theta}\\
1 & 0 & 1 &\ddots &  & 0\\
0 & 1 &\ddots &\ddots &\ddots &\vdots\\
\vdots &\ddots &\ddots &\ddots & 1 & 0\\
0 & &\ddots & 1 & 0 & 1\\
\e^{i\theta} & 0 &\cdots & 0 & 1 & 0
\end{pmatrix}.
\end{equation*}
These two ingredients lead to the self-adjoint operator $H_0(\theta)$ 
acting on $\ell^2\big(\N;\C^N\big)$ as
$$
H_0(\theta):=\Delta_\NN \otimes 1_N+ A^\theta
$$
with $1_N$ the $N \times N $ identity matrix.

This operator can be viewed as a discrete magnetic adjacency operator, see Remark \ref{rem:magnetic}. It is also the one which appears in a direct integral decomposition of 
an operator acting on $\ell^2(\N\times \Z)$ through a Bloch-Floquet
transformation, see \cite[Sec.~2]{NRT}. We do not emphasise this decomposition 
in the present work, but concentrate on each individual operator living on the fibers
of a direct integral.

The perturbed operator $H(\theta)$ describing the discrete quantum model is then given by
\begin{equation*}
H(\theta):=H_0(\theta)+V,
\end{equation*}
where $V$ is the multiplication operator by a nonzero, matrix-valued function
with support on $\{0\}\in \N$. In other words, there exists a nonzero function
$v:\{1,\dots,N\}\to\R$ such that for $\psi\in \ell^2\big(\N;\C^N\big)$, $n\in\N$, and $j\in \{1, \dots, N\}$
one has
$$
\big(H(\theta)\psi\big)_j(n)=\big(H_0(\theta)\psi\big)_j(n)+\delta_{0,n}\;\!v(j)\;\!\psi_j(0),
$$
with $\delta_{0,n}$ the Kronecker delta function.

\begin{Remark}\label{rem:magnetic}
$H_0(\theta)$ can be interpreted as a magnetic adjacency operator on the Cayley graph of the semi-group $\N\times (\Z/N \Z)$ with respect to a magnetic field constant in the $\N$ direction. In particular, it correspond to choosing a magnetic potential supported only on edges of the form $\big((n,N),(n,1)\big)$ for any $n\in \N$. In this representation, the perturbation $V$ corresponds
to a multiplicative perturbation with support on $\{0\}\times (\Z/N \Z)$.
\end{Remark}

We now set
$$
\h:=\ltwo\big([0,\pi),\tfrac{\d\omega}\pi;\C^N\big)
$$
and define a transformation in the $\N$-variable
$$
\G: \ell^2\big(\N;\C^N\big) \to\h
$$
given for $\psi \in \ell^2\big(\N;\C^N\big)_{\rm fin}$ (the finitely supported functions), $\omega\in [0,\pi)$, and $j\in \{1,\dots,N\}$ by
$$
\big(\G \psi\big)_j(\omega)
:=2^{1/2}\sum_{n\ge1}\cos(n\;\!\omega)\psi_j(n)+\psi_j(0).
$$
The transformation $\G$ extends to a unitary operator from
$\ell^2\big(\N;\C^N\big)$ to $\h$, which we denote by the same
symbol, and a direct computation gives the equality
$$
\G\;\!H_0(\theta)\;\!\G^*= H^\theta_0
\quad\hbox{with}\quad
H^\theta_0:=2\cos(\Omega)\otimes 1_N+A^\theta
$$
with the operator $2\cos(\Omega)$ of multiplication by the function
$[0,\pi)\ni\omega\mapsto2\cos(\omega)\in\R$.

Through the same unitary transformation, the operator $H(\theta)$ 
is unitarily equivalent to 
\begin{equation*}
H^\theta:=2\cos(\Omega)\otimes 1_N+A^\theta+\diag(v)P_0
\end{equation*}
with
\begin{equation*}
\big(\diag(v)\;\!\f\big)_j
:=v(j)\;\!\f_j
\quad\hbox{and}\quad
\big(P_0\;\!\f\big)_j
:=\int_0^\pi\f_j(\omega)\tfrac{\d\omega}\pi
\end{equation*}
for $\f\in\h$, $j\in\{1,\dots,N\}$.
Observe that the term $\diag(v)P_0$ corresponds to a finite rank
perturbation, and therefore $H^\theta$ and $H_0^\theta$
differ only by a finite rank operator.

Let us now move to spectral results. 
A direct inspection shows that the matrix $A^\theta$ has eigenvalues
$$
\lambda_j^\theta:=2\cos\left(\frac{\theta+2\pi\;\!j}N\right),\quad j\in\{1,\dots,N\},
$$
with corresponding eigenvectors $\xi_j^\theta\in\C^N$ having components
$\big(\xi_j^\theta\big)_k:=\e^{i(\theta+2\pi j)k/N}$, $j,k\in\{1,\dots,N\}$. Using the
notation $\P_j^\theta\equiv \tfrac{1}{N}|\xi^\theta_j\rangle \langle \xi^\theta_{j}|$ for the orthogonal projection associated to $\xi_j^\theta$, one has 
$A^\theta=\sum_{j=1}^N\lambda_j^\theta\;\!\P_j^\theta$.

The next step consists in exhibiting the spectral representation of $H_0^\theta$.
For that purpose, we first define for
$\theta\in(0,\pi)$ and $j\in\{1,\dots,N\}$ the sets
\begin{equation*}
I_j^\theta:=\big(\lambda_j^\theta-2,\lambda_j^\theta+2\big)
\quad\hbox{and}\quad
I^\theta:=\cup_{j=1}^NI^\theta_j,
\end{equation*}
with $\lambda_j^\theta$ the eigenvalues of $A^\theta$. 
Also, we consider for $\lambda\in I^\theta$ the fiber Hilbert space
$$
\Hrond^\theta(\lambda)
:=\sp\big\{\P^\theta_j\C^N\mid\hbox{$j\in\{1,\dots,N\}$ such that
$\lambda\in I_j^\theta$}\big\}\subset\C^N,
$$
and the corresponding direct integral Hilbert space
$$
\Hrond^\theta:=\int_{I^\theta}^\oplus\Hrond^\theta(\lambda)\;\!\d\lambda.
$$
Then, the map $\F^\theta:\h\to\Hrond^\theta$ 
acting on $\f\in\h$ and for a.e.~$\lambda\in I^\theta$ as
$$
\big(\F^\theta \f\big)(\lambda)
:=\pi^{-1/2}\sum_{\{j\mid\lambda\in I_j^\theta\}}
\big(4-\big(\lambda-\lambda_j^\theta\big)^2\big)^{-1/4}\;\!\P^\theta_j
\f\left(\arccos\left(\tfrac{\lambda-\lambda_j^\theta}2\right)\right).
$$
is unitary.
In addition, $\F^\theta$ diagonalises the Hamiltonian $H_0^\theta$, namely for all
$\zeta\in\Hrond^\theta$ and a.e. $\lambda\in I^\theta$ one has
\begin{equation*}
\big(\F^\theta H_0^\theta\;\!(\F^{\theta})^*\zeta\big)(\lambda)
=\lambda\;\!\zeta(\lambda)
=\big(X^\theta\zeta\big)(\lambda),
\end{equation*}
with $X^\theta$ the (bounded) operator of multiplication by the variable in
$\Hrond^\theta$.
One infers  that
$H_0^\theta$ has purely absolutely continuous spectrum equal to
\begin{equation*}
\textstyle\sigma(H_0^\theta)
=\overline{\Ran(X^\theta)}
=\overline{I^\theta}
=\big[\big(\min_j\lambda_j^\theta\big)-2,\big(\max_j\lambda_j^\theta\big)+2\big]
\subset[-4,4].
\end{equation*}
Note that we use the notation  $\overline{I}$ for the closure of a set $I\subset \R$.

The spectral representation of $H^\theta_0$ leads also naturally to the
notion of thresholds: these real values correspond to a change
of multiplicity of the spectrum.
Clearly, the set $\T^\theta$ of thresholds for the operator $H_0^\theta$
is given by
\begin{equation}\label{eq:thresh}
\T^\theta:=\big\{\lambda_j^\theta\pm2\mid j\in\{1,\dots,N\}\big\}.
\end{equation}

\begin{Remark}\label{rem:0pi}
By looking at these thresholds, one observes that there is a qualitative change between
these thresholds for $\theta \in (0,\pi)$ and for $\theta=0$ or $\theta = \pi$.
Indeed, for $\theta \in \{0,\pi\}$, some of these thresholds coincide, while it is not the case for
$\theta \in (0,\pi)$. This multiplicity different from $1$ has
several implications in the following constructions, and for simplicity we have decided
not to consider these pathological cases here. 
Note also that the matrices obtained for $\theta \in (\pi,2\pi)$ are unitarily equivalent to some matrices with $\theta\in (0,\pi)$. Again for simplicity, we restrict our attention
to $\theta \in (0,\pi)$ only.
\end{Remark}

The main spectral result for $H^\theta$ has been presented in \cite[Prop.~1.2]{NRT},
which we recall when necessary, see Proposition \ref{proposition_kernel}. About scattering theory, since the difference $H^\theta-H^\theta_0$ is a finite rank
operator we observe that the wave operators
$$
W_{\pm}^\theta:=\slim_{t\to\pm\infty}\e^{itH^\theta}\e^{-itH^\theta_0}
$$
exist and are complete, see  \cite[Thm~X.4.4]{Kat95}. As a consequence, the scattering operator
$$
S^\theta:=(W_+^\theta)^*W_-^\theta
$$
is a unitary operator in $\h$ commuting with $H^\theta_0$, and thus $S^\theta$ is
decomposable in the spectral representation of $H^\theta_0$, that is
for $\zeta\in\Hrond^\theta$ and  a.e.~$\lambda\in I^\theta$, one has
$$
\big(\F^\theta S^\theta (\F^\theta)^*\zeta\big)(\lambda)=S^\theta(\lambda)\zeta(\lambda),
$$
with the scattering matrix $S^\theta(\lambda)$ a unitary operator in $\Hrond^\theta(\lambda)$.

For $j,j'\in\{1,\dots,N\}$ and for
$
\lambda\in\big(I^\theta_j\cap I^\theta_{j'}\big)
\setminus\big(\T^\theta\cup\sigma_{\rm p}(H^\theta)\big)
$
let us define the channel scattering matrix
$S^\theta(\lambda)_{jj'}:=\P^\theta_jS^\theta(\lambda)\P^\theta_{j'}$
and consider the map
$$
\big(I^\theta_j\cap I^\theta_{j'}\big)
\setminus\big(\T^\theta\cup\sigma_{\rm p}(H^\theta)\big)
\ni\lambda\mapsto S^\theta(\lambda)_{jj'}
\in\B\big(\P_{j'}^\theta\C^N;\P_j^\theta\C^N\big).
$$
The continuity of the scattering matrix at embedded eigenvalues has been shown in 
\cite[Thm~3.10]{NRT}, while its behaviour at thresholds has been studied in 
\cite[Thm~3.9]{NRT}.
This latter result can be summarised as follows: For each $\lambda\in\T^\theta$, 
a channel can already be opened at the
energy $\lambda$ (in which case the existence and the equality of the
limits from the right and from the left is proved), it can open at the energy $\lambda$ (in
which case only the existence of the limit from the right is proved), or it can
close at the energy $\lambda$ (in which case only the existence of the
limit from the left is proved).

Let us finally return to the wave operator $W_-^\theta$, which will be further investigated in the sequel. By using the stationary approach of scattering theory
and by looking at the representation of the wave operator inside
the spectral representation of $H_0^\theta$, the operator $W_-^\theta$
can be explicitely computed, modulo compact operators. More precisely
if we set $\K(\h)$ for the set of compact operators on $\h$, then
the following statement has been proved in \cite[Thm~1.5]{NRT}~:

\begin{Theorem}\label{thm_formula_theta}
For any $\theta\in(0,\pi)$, one has the equality
$$
W_-^\theta-1
=\tfrac12\big(1-\tanh(\pi\Dt)-i\cosh(\pi\Dt)^{-1}\tanh(\Xt)\big)(S^\theta-1)
+\Kt^\theta,
$$
with $\Kt^\theta \in\K(\h)$ and where  $\Xt$ and $\Dt$ are representations
of the canonical position and momentum operators in the Hilbert space $\h$.
\end{Theorem}

Let us emphasise that this result is at the root of the subsequent investigations.

\section{The \texorpdfstring{$C^*$}{C*}-algebraic framework}\label{sec:C*}
\setcounter{equation}{0}

In the section we construct the algebraic framework and a $C^*$-algebra
which is going to contain the operator $W_-^\theta$ for any fixed $\theta \in (0,\pi)$. 
We also compute the quotient of this $C^*$-algebra by the ideal of compact operators.
Most of the notations are directly borrowed from Section \ref{sec:themodel}. 
 
We firstly exhibit an algebra $\AA^\theta$ of multiplication operators acting on $\Hrond^\theta$.

\begin{Definition}
Let $\AA^\theta\subset \B\big(\Hrond^\theta\big)$ with  $a\in \AA^\theta$ if for any $\zeta \in \Hrond^\theta$ and $\lambda \in I^\theta$ one has
$$
\big[a\zeta\big](\lambda)\equiv \big[a(X^\theta)\zeta\big](\lambda)= a(\lambda) \zeta(\lambda)
$$
and for $j,j'\in \{1,\dots,N\}$ the maps
\begin{equation}\label{eq:mult}
I_j^\theta \cap I_{j'}^\theta \ni \lambda \mapsto a_{jj'}(\lambda) :=
\P_j^\theta a(\lambda)\P_{j'}^\theta \in \B\big(\Hrond^\theta(\lambda)\big) \subset \B(\C^N)
\end{equation}
belongs to $C_0\left(I_j^\theta \cap I_{j'}^\theta; \B(\C^N)\right)$ if $j\neq j'$ while it belongs to $C\left(\overline{I_j^\theta}; \B(\C^N)\right)$ if $j=j'$.
\end{Definition}

In other words, the function defined in \eqref{eq:mult}
is a continuous function vanishing at the boundaries of the intersection $I_j^\theta \cap I_{j'}^\theta$ when these two intervals are not equal, and has limits at the boundaries of these intervals when they coincide. 

\begin{Remark}\label{rem:misleading}
In the sequel, we shall use the notation $a_{jj'}$ for $\P_j^\theta a \P_{j'}^\theta$
even if this notation is slightly misleading. Indeed, $a_{jj'}(\lambda)$ exists if and only if 
$\lambda \in  \overline{I_j^\theta \cap I_{j'}^\theta}$. 
If $\lambda \not \in 
\overline{I_j^\theta}$ or if $\lambda \not \in\overline{I^\theta_{j'}}$, then the expression containing
$a_{jj'}(\lambda)$ should be interpreted as $0$.
More precisely, with the standard bra-ket notation, we shall write
$a_{jj'}=\a_{jj'}\otimes \;\!\tfrac{1}{N}|\xi^\theta_j\rangle \langle \xi^\theta_{j'}|$,
where $\a_{jj'}$ denotes the multiplication operator by a function belonging to  
$C_0\left(I_j^\theta \cap I_{j'}^\theta; \C\right)$ if $j\neq j'$, while it belongs to $C\big(\overline{I_j^\theta}; \C\big)$ if $j=j'$.
Thus, with these notations and for any $\lambda \in I^\theta$ one has
$$
a(\lambda)=\sum_{\{j,j'\mid \lambda\in I^\theta_j, \lambda \in I^\theta_{j'}\}}
\a_{jj'}(\lambda) \otimes \;\!\tfrac{1}{N}|\xi^\theta_j\rangle \langle \xi^\theta_{j'}|
$$
By the above abuse of notation, we shall conveniently write
$a=\sum_{j,j'}
\a_{jj'}\otimes \;\!\frac{1}{N}|\xi^\theta_j\rangle \langle \xi^\theta_{j'}|$.
\end{Remark}

The interest of the algebra $\AA^\theta$ comes from the following statement, which can be directly inferred from \cite[Thms 3.9 \& 3.10]{NRT}\;\!:

\begin{Proposition}\label{prop_S}
The multiplication operator defined by the map $\lambda \to S^\theta(\lambda)$ belongs to $\AA^\theta$. 
\end{Proposition}

For subsequent constructions, we need to know that some specific functions also belong
to $\AA^\theta$. For this, we recall the definition of a useful unitary map, introduced in
\cite[Sec.~4.1]{NRT}.
For $j\in\{1,\dots,N\}$ we define the unitary map
$\V_j^\theta:\ltwo(I_j^\theta)\to\ltwo(\R)$ given on $\zeta\in\ltwo(I_j^\theta)$ and for $s\in\R$ by
\begin{equation}\label{eq:Vjt}
\big[\V_j^\theta\zeta\big](s)
:=\tfrac{2^{1/2}}{\cosh(s)}\;\!\zeta\big(\lambda_j^\theta+2\tanh(s)\big).
\end{equation}
Its adjoint $(\V^{\theta}_j)^*:\ltwo(\R)\to\ltwo(I_j^\theta)$ is then given on $f\in\ltwo(\R)$ and for  $\lambda\in I_j^\theta$ by
\begin{equation}\label{eq:Vjt*}
\big[(\V^{\theta}_j)^*f\big](\lambda)
=\left(\tfrac2{4-(\lambda-\lambda_j^\theta)^2}\right)^{1/2}
f\left(\arctanh\left(\tfrac{\lambda-\lambda_j^\theta}2\right)\right).
\end{equation}
Based on this, we define the unitary map
$\V^\theta:\Hrond^\theta\to\ltwo(\R;\C^N)$  acting on $\zeta\in\Hrond^\theta$ as
$$
\V^\theta\zeta
:=\sum_{j=1}^N\big(\V^\theta_j\otimes\P_j^\theta\big)\zeta|_{I^\theta_j},
$$
with adjoint $(\V^{\theta})^*:\ltwo(\R;\C^N)\to\Hrond^\theta$ acting on 
$f\in\ltwo(\R;\C^N)$ and for $\lambda\in I^\theta$ as
$$
\big[(\V^{\theta})^*f\big](\lambda)
=\sum_{\{j\mid\lambda\in I^\theta_j\}}
\left[\big((\V^{\theta}_j)^*\otimes\P_j^\theta\big)f\right](\lambda).
$$

In the sequel, let $X$ denote the self-adjoint operator of multiplication by the variable in 
$\ltwo(\R)$.
We are now ready to prove the following lemma:

\begin{Lemma}\label{lem_tanh}
For  $\zeta\in\Hrond^\theta$ and $\lambda \in I^\theta$ one has
\begin{equation*}
\big[\big(\V^\theta\big)^* \big(\tanh(X)\otimes 1_N \big)\V^\theta \zeta\big](\lambda)
= \sum_{\{j\mid\lambda\in I^\theta_j\}}\left[\tfrac{X^\theta-\lambda^\theta_j}{2}\otimes \P^\theta_j \zeta\right](\lambda).
\end{equation*}
\end{Lemma}

\begin{proof}
From the definitions of the unitary maps one infers that
\begin{align*}
&\big[\big(\V^\theta\big)^* \big(\tanh(X)\otimes 1_N \big)\V^\theta \zeta\big](\lambda) \\
& = \sum_{\{j\mid\lambda\in I^\theta_j\}}
\left[\big((\V^{\theta}_j)^*\otimes\P_j^\theta\big)\big(\tanh(X)\otimes 1_N \big)\V^\theta \zeta\right](\lambda) \\
& = \sum_{\{j\mid\lambda\in I^\theta_j\}}
\left(\tfrac2{4-(\lambda-\lambda_j^\theta)^2}\right)^{1/2}
\left(\tfrac{\lambda-\lambda_j^\theta}2\right) \P_j^\theta
\big[\V^\theta \zeta\big]\left(\arctanh\left(\tfrac{\lambda-\lambda_j^\theta}2\right)\right) \\
& = \sum_{\{j\mid\lambda\in I^\theta_j\}}
\left(\tfrac2{4-(\lambda-\lambda_j^\theta)^2}\right)^{1/2}
\left(\tfrac{\lambda-\lambda_j^\theta}2\right) 
\big[\V^\theta_j \otimes \P_j^\theta \zeta\big]\left(\arctanh\left(\tfrac{\lambda-\lambda_j^\theta}2\right)\right) \\
& = \sum_{\{j\mid\lambda\in I^\theta_j\}}
\tfrac{2^{-1/2}}{\left(1-\Big(\tfrac{\lambda-\lambda_j^\theta}{2}\Big)^2\right)^{1/2}}
\left(\tfrac{\lambda-\lambda_j^\theta}2\right) \frac{2^{1/2}}{\cosh\left(\arctanh\left(\tfrac{\lambda-\lambda_j^\theta}2\right)\right)}
\big[\P_j^\theta \zeta\big](\lambda)  \\
& =  \sum_{\{j\mid\lambda\in I^\theta_j\}}
\left(\tfrac{\lambda-\lambda_j^\theta}2\right) 
\big[\P_j^\theta \zeta\big](\lambda), 
\end{align*}
where the relation $\frac{1}{\cosh(s)^2}=1-\tanh(s)^2$ has been used for the last equality.
This leads directly to the statement.
\end{proof}

As a consequence of this lemma one easily infers that
$\big(\V^\theta\big)^* \big(\tanh(X)\otimes 1_N \big)\V^\theta$ belongs to the algebra $\AA^\theta$ introduced above. In fact, in the matrix formulation mentioned in Remark 
\ref{rem:misleading}, this operator corresponds to a diagonal multiplication operator.

Our next aim is to introduce a $C^*$-algebra containing the wave operators $W_\pm^\theta$.
For that purpose, we introduce in the Hilbert space $\ltwo(\R;\C^N)$ a specific family of convolution operators. More precisely, let $D$ stand for the self-adjoint realization of the operator $-i\frac{\d}{\d x}$ in $\ltwo(\R)$. Clearly, $X$ and $D$ satisfy the canonical commutation relation. By functional calculus, we then define the bounded convolution operator $\eta(D)$ with $\eta\in C_0\big([-\infty,+\infty)\big)$, the algebra of continuous functions on $\R$ having a limit at $-\infty$ and vanishing at $+\infty$.
In $\ltwo(\R;\C^N)$ we finally introduce the family of operators 
$\eta(D)\otimes 1_N$
with $\eta\in  C_0\big([-\infty,+\infty)\big)$. 
These operators naturally generate a $C^*$-algebra $\CC$ isomorphic to  $C_0\big([-\infty,+\infty)\big)$.
Clearly, the operators $\frac{1}{2}\big(1-\tanh(\pi D)\big)\otimes 1_N$ and $\cosh(\pi D)^{-1}\otimes 1_N$ belong to 
$\CC$. These two operators will subsequently play an important role.

For $\eta\in  C_0\big([-\infty,+\infty)\big)$, let us now look at the image of 
$\eta(D)\otimes 1_N$ in $\Hrond^\theta$. 
For this, we recall that $\H^1(\R)$ denotes the first Sobolev space on $\R$. 

\begin{Lemma}
\begin{enumerate}
\item[(i)] For $j\in \{1,\dots,N\}$ the operator
$D^\theta_j:=\big(\V_j^\theta\big)^* D\V^\theta_j$ is self-adjoint on $(\V_j^\theta\big)^* \H^1(\R)$, 
and the following equality holds:
$$
D^\theta_j=2\bigg(1-\Big(\tfrac{X^\theta_j-\lambda_j^\theta}{2}\Big)^2\bigg)\left(-i\tfrac{\d}{\d \lambda}\right) 
+i\Big(\tfrac{X^\theta_j-\lambda^\theta_j}{2}\Big)
$$
where $X^\theta_j$ denotes the operator of multiplication by the variable in $\ltwo(I_j^\theta)$,
\item[(ii)] For any $\eta\in  C_0\big([-\infty,+\infty)\big)$ we set
\begin{equation}\label{eq_Dtheta}
\eta(D^\theta):=\big(\V^\theta\big)^* \big(\eta(D)\otimes 1_N \big)\V^\theta, 
\end{equation}
and for $\zeta \in \Hrond^\theta$ and for $\lambda\in I^\theta$ one has
$$
\big[\eta(D^\theta) \zeta\big](\lambda) 
= \sum_{\{j\mid\lambda\in I^\theta_j\}} \big[\big(\eta(D^\theta_j)\otimes \P_j^\theta\big)\zeta\big](\lambda).
$$
\end{enumerate}
\end{Lemma}

\begin{proof}
i) The self-adjointness of $D^\theta_j$ on $(\V_j^\theta\big)^* \H^1(\R)$ directly follows from the self-adjointness of $D$ on $\H^1(\R)$.
Then, for $\zeta\in C^\infty_{\rm c}\big(I_j^\theta\big)$ one has
\begin{align*}
&\big[\big(\V_j^\theta\big)^* D\V^\theta_j \zeta \big](\lambda) \\
& = \left(\tfrac2{4-(\lambda-\lambda_j^\theta)^2}\right)^{1/2} \big[ D\V^\theta_j \zeta \big]
\left(\arctanh\left(\tfrac{\lambda-\lambda_j^\theta}2\right)\right) \\
& = -i \left(\tfrac2{4-(\lambda-\lambda_j^\theta)^2}\right)^{1/2} \big[ \V^\theta_j \zeta \big]'
\left(\arctanh\left(\tfrac{\lambda-\lambda_j^\theta}2\right)\right) \\
& =  -i \left(\tfrac2{4-(\lambda-\lambda_j^\theta)^2}\right)^{1/2}\left[\tfrac{2^{1/2}}{\cosh(\cdot)}\zeta \big(\lambda_j^\theta+2 \tanh(\cdot)\big)\right]' \left(\arctanh\left(\tfrac{\lambda-\lambda_j^\theta}2\right)\right) \\
& = -i \tfrac{1}{\left(1-\Big(\tfrac{\lambda-\lambda_j^\theta}{2}\Big)^2\right)^{1/2}}
\Big[\tfrac{-\tanh(\cdot)}{\cosh(\cdot)}\zeta \big(\lambda_j^\theta+2 \tanh(\cdot)\big)
\\
& \qquad \qquad  \qquad + 2\tfrac{1-\tanh(\cdot)^2}{\cosh(\cdot)}\zeta'\big(\lambda_j^\theta +2 \tanh(\cdot)\big) \Big]\left(\arctanh\left(\tfrac{\lambda-\lambda_j^\theta}2\right)\right) \\
& = i \tfrac{\lambda-\lambda_j^\theta}{2} \zeta(\lambda) 
-2i\bigg(1-\left(\tfrac{\lambda-\lambda_j^\theta}{2}\right)^2\bigg)\;\!\zeta'(\lambda),
\end{align*}
which leads directly to the statement.

ii) For $\zeta\in \Hrond^\theta$ and $\lambda \in I^\theta$ one has
\begin{align*}
\big[\big(\V^\theta\big)^* \big(\eta(D)\otimes 1_N \big)\V^\theta \zeta\big](\lambda) 
& = \sum_{\{j\mid\lambda\in I^\theta_j\}}
\left[\big((\V^{\theta}_j)^*\otimes\P_j^\theta\big)\big(\eta(D)\otimes 1_N \big)\V^\theta \zeta\right](\lambda) \\
& = \sum_{\{j\mid\lambda\in I^\theta_j\}}
\left[\big(\V^{\theta}_j)^*\eta(D)\otimes \P_j^\theta \big)\V^\theta \zeta\right](\lambda) \\
& = \sum_{\{j\mid\lambda\in I^\theta_j\}}
\left[\big(\V^{\theta}_j)^*\eta(D)\V^\theta_j \otimes \P_j^\theta \big) \zeta\right](\lambda) \\
& = \sum_{\{j\mid\lambda\in I^\theta_j\}}
\left[\big(\eta(D_j^\theta) \otimes \P_j^\theta \big) \zeta\right](\lambda),
\end{align*}
which corresponds to the statement.
\end{proof}

Let us now introduce our main $C^*$-algebra. We recall that $\eta(D^\theta)$ has been introduced in \eqref{eq_Dtheta}. 
We set
\begin{equation*}
\EE^\theta:=C^*\Big(\eta(D^\theta)\;\!a+b\mid \eta \in   C_0\big([-\infty,+\infty)\big),  a\in \AA^\theta, b\in \K\big(\Hrond^\theta\big)\Big)^+
\end{equation*}
which is a $C^*$-subalgebra of $\B\big(\Hrond^\theta\big)$
containing the ideal $\K\big(\Hrond^\theta\big)$ of compact operators on $\Hrond^\theta$.
Here, the exponent $+$ means that $\C$ times the identity of $\B\big(\Hrond^\theta\big)$ have been added to the algebra, turning it into a unital $C^*$-algebra.
Observe that for a typical element
$\eta(D^\theta)\;\!a$ with $\eta \in C_0\big([-\infty,+\infty)\big)$ and $a\in \AA^\theta$,
and for $\zeta\in \Hrond^\theta$ and $\lambda \in I^\theta$, one has
\begin{align}
\nonumber \big[\eta(D^\theta)\;\!a \zeta\big](\lambda) 
&=  \sum_{\{j\mid\lambda\in I^\theta_j\}} \big[\big(\eta(D^\theta_j)\otimes \P_j^\theta\big)a \zeta\big](\lambda) \\
\nonumber  & = \sum_{\{j\mid\lambda\in I^\theta_j\}}
\sum_{j'=1}^N \big[\eta(D^\theta_j) \;\!a_{jj'} \P_{j'}^\theta\zeta\big](\lambda) \\
\nonumber &=  \sum_{\{j\mid\lambda\in I^\theta_j\}} \sum_{j'=1}^N
\big[\big(\eta(D^\theta_{j})\;\!\a_{jj'}\big)\otimes \;\!\tfrac{1}{N}|\xi^\theta_j\rangle \langle \xi^\theta_{j'}|\zeta\big](\lambda)\\
\label{eq:sum} &=  \sum_{j,j'}
\big[\big(\eta(D^\theta_{j})\;\!\a_{jj'}\big)\otimes \;\!\tfrac{1}{N}|\xi^\theta_j\rangle \langle \xi^\theta_{j'}|\zeta\big](\lambda),
\end{align}
where the notation introduced in Remark \ref{rem:misleading} has been used.

Our main interest in this $C^*$-algebra is based on the following result, which is a direct
consequence of Theorem \ref{thm_formula_theta}.

\begin{Proposition}\label{prop:affiliation}
For any $\theta \in (0,\pi)$, the wave operators $\F^\theta W_-^\theta(\F^\theta)^*$ belongs to $\EE^\theta$.
\end{Proposition}

\begin{proof}
Let us start by recalling the main formula for the wave operator valid for $\theta \in (0,\pi)$\;\!:
$$
W_-^\theta-1
=\tfrac12\big(1-\tanh(\pi\Dt)-i\cosh(\pi\Dt)^{-1}\tanh(\Xt)\big)(S^\theta-1)
+\Kt^\theta,
$$
with $\Kt^\theta \in\K(\h)$. By looking at the r.h.s.~through the unitary conjugation by $\F^\theta$, one ends up with the expression in $\Hrond^\theta$\;\!:
\begin{equation}\label{eq_W_Hrond}
\tfrac12\Big(1-\tanh(\pi D^\theta)-i\cosh(\pi D^\theta)^{-1}
\big(\V^\theta\big)^* \big(\tanh(X)\otimes 1_N \big)\V^\theta
\Big)(S^\theta(X^\theta)-1) +k^\theta,
\end{equation}
with $k^\theta\in \K\big(\Hrond^\theta\big)$.

Clearly, the functions $1-\tanh(\pi\cdot)$ and $\cosh(\pi\cdot)^{-1}$ belong to
$C_0\big([-\infty,+\infty)\big)$.
As already mentioned in Proposition \ref{prop_S}, the map $\lambda \mapsto S^\theta(\lambda)$ belongs $\AA^\theta$, and the same is true for the constant function $1$. As a consequence, the map $\lambda \mapsto S^\theta(\lambda)-1$ belongs to $\AA^\theta$, or equivalently the multiplication operator $S^\theta(X^\theta)-1$ belongs to $\AA^\theta$.
Furthermore and as mentioned after Lemma \ref{lem_tanh}, the operator
$\big(\V^\theta\big)^* \big(\tanh(X)\otimes 1_N \big)\V^\theta$ also belongs to the 
algebra $\AA^\theta$.
Thus, the leading term in \eqref{eq_W_Hrond} belongs to $\EE^\theta$ by inspection of its different factors, and the remainder term $k^\theta$ also belongs to $\EE^\theta$ because it is a compact contribution.
\end{proof}

Our next task is to compute the quotient of the algebra $\EE^\theta$ by the set of compact operators. 
For that purpose, we firstly show that if some restrictions are imposed on the functions $\a_{jj'}$ and $\eta$ appearing in \eqref{eq:sum}, 
then the corresponding operator is compact. For this we introduce an ideal in $\AA^\theta$, namely
$$
\AA^\theta_0 :=\Big\{ a\in \AA^\theta \mid
a_{jj}\in C_0\left(I_j^\theta; \B(\C^N)\right)\hbox{ for all } j\in \{1,\dots,N\}\Big\}.
$$

\begin{Lemma}\label{lem:compact}
For $a\in \AA^\theta_0$ and $\eta \in C_0(\R)$ the operator $\eta(D^\theta)a$ belongs
to $\K\big(\Hrond^\theta\big)$.
\end{Lemma}

\begin{proof}
Let us first observe that the condition $a\in \AA^\theta_0$ means that all multiplication operators
$\a_{jj'}$ introduced in Remark \ref{rem:misleading} are defined by functions in $C_0\left(I_j^\theta \cap I_{j'}^\theta; \C\right)$, for all $j, j'\in \{1,\dots,N\}$. 
Thus, according to \eqref{eq:sum} it is enough to show that each operator
$\eta(D^\theta_{j})\,\!\a_{jj'}: \ltwo(I_{j'}^\theta)\to \ltwo(I_j^\theta)$ is compact. However, since $\supp \a_{jj'}\subset \overline{I_j^\theta \cap I_{j'}^\theta}$, 
it is simpler to show that $\eta(D^\theta_{j})\,\!\a_{jj'}\in \K\big(\ltwo(I_{j}^\theta)\big)$.
Now, if we apply the unitary transforms defined in \eqref{eq:Vjt} and \eqref{eq:Vjt*}
one gets
\begin{equation}\label{eq:compact}
\V_{j}^\theta\;\! \eta(D^\theta_{j})\,\!\a_{jj'}\;\!\big(\V_{j}^\theta\big)^*
= \eta(D)\;\!\a_{jj'}\big(\lambda_{j}^\theta + 2\tanh(X)\big).
\end{equation}
Since $\eta$ and $\a_{jj'}\big(\lambda_{j}^\theta + 2\tanh(\cdot)\big)$ belong to $C_0(\R)$,
the r.h.s.~of \eqref{eq:compact} is known to be a compact operator in $\ltwo(\R)$.
Thus, each summand in \eqref{eq:sum} is unitarily equivalent to a compact operator, which means
that $\eta(D^\theta)a$ is compact.
\end{proof}

Let us supplement the previous result with a compactness result about commutators.

\begin{Lemma}
For any $\eta\in C_0\big([-\infty,\infty)\big)$ and $a\in \AA^\theta$, the commutator
$[\eta(D^\theta),a]$ belongs to $\K\big(\Hrond^\theta\big)$.
\end{Lemma}

\begin{proof}
It suffices to show that $F^\theta:=D^\theta\big(1+(D^\theta)^2\big)^{-1/2}$ has compact commutator with elements of $\AA^\theta$. For if this is true, then any polynomial in $F^\theta$ compactly commutes with $\AA^\theta$.
If $\eta\in C_0\big([-\infty,\infty)\big)$ then there exists a unique $h\in C_0\big([-1,1)\big)$ such that $\eta(D^\theta)=h(F^\theta)$. By Stone-Weierstrass theorem, any $h \in C_0\big([-1,1)\big)$ can be approximated by a polynomial on $(-1,1)$. Then, one concludes the proof by recalling that the set of compact operators is norm closed.

Now, with the notation of Remark \ref{rem:misleading} and for any $a\in \AA^\theta$, 
let us define $\diag(a)$ by
$$
\big[\diag(a)\big](\lambda):=
\sum_{\{j\mid \lambda\in I^\theta_j\}}
\a_{jj}(\lambda)\otimes \;\!\P^\theta_j, \qquad \forall \lambda \in I^\theta.
$$
We also consider the diagonal multiplication operator $\Lambda\in \AA^\theta_0$ 
given for any $\lambda \in I^\theta$ by 
$$
\Lambda^\theta(\lambda):=\sum_{\{j\mid \lambda\in I^\theta_j\}}
\Lambda^\theta_j(\lambda)\otimes \;\!\P^\theta_j
$$
with $\Lambda^\theta_j(\lambda):=1-\Big(\tfrac{\lambda-\lambda^\theta_j}{2}\Big)^2$.
For $\epsilon>0$ small, let $\T^\theta_\epsilon$ be an open $\epsilon$-neighbourhood of $\T^\theta$, where the set of thresholds $\T^\theta$ has been defined in \eqref{eq:thresh}. 
We then set $\AA^\theta_\epsilon$ for the set of $a\in \AA^\theta$ 
satisfying $\supp\big(a-\diag(a)\big) \subset I^\theta\setminus \T^\theta_\epsilon$.
For $a\in \AA^\theta_\epsilon$ we define
$$
\rho^\theta(a):={\rm diag}(a)+\Lambda^\theta\big(a-\diag(a)\big)\big(\Lambda^\theta\big)^{-1}.
$$
With the notations introduced before it means that for any $\lambda \in I^\theta$ one has
\begin{align*}
\big[\rho^\theta(a)\big](\lambda)=&
\sum_{\{j\mid \lambda\in I^\theta_j\}}
\a_{jj}(\lambda)\otimes \;\!\P^\theta_j \\
& + \sum_{\{j\neq j'\mid \lambda\in I^\theta_j, \lambda \in I^\theta_{j'}\}}
\Lambda^\theta_j(\lambda)\a_{jj'}(\lambda)\big(\Lambda^\theta_{j'}(\lambda)\big)^{-1}\otimes \;\!\tfrac{1}{N}|\xi^\theta_j\rangle \langle \xi^\theta_{j'}|.
\end{align*}
One observes that $\rho^\theta$ is an algebra homomorphism on the subalgebra $\AA^\theta_\epsilon$, 
but it is not a $*$-homomorphism. With a grain of salt, the expression $\rho^\theta(a)$
could simply be written $\Lambda^\theta a \big(\Lambda^\theta\big)^{-1}$, but the diagonal
elements have to be suitably interpreted (the cancellation of two factors before the evaluation
at thresholds).

Let us also consider the subalgebra $\AA^\theta_{\epsilon,1}$ of $\AA^\theta_\epsilon$ of those $a$ whose derivatives  $a'$ exist, are continuous and satisfy
$\supp(a') \subset I^\theta\setminus \T^\theta_\epsilon$.
Our interest in this algebra is that any element of $\AA^\theta$ is a norm limit of elements in $\AA^\theta_{\epsilon,1}$ as $\epsilon \searrow 0$.
If we define the additional diagonal multiplication operator
$$
\Omega^\theta(\lambda)=
\sum_{\{j\mid \lambda\in I^\theta_j\}}
i\Big(\tfrac{\lambda-\lambda^\theta_j}{2}\Big)\otimes \;\!\P^\theta_j, \qquad \forall \lambda \in I^\theta,
$$
then one readily checks that for $a\in\AA^\theta_{\epsilon,1}$ one has in the form sense
on the domain of $D^\theta$
\begin{align*}
D^\theta a - \rho^\theta(a) D^\theta
&=-2i\Lambda^\theta a'+\Omega^\theta a-\rho^\theta(a)\Omega^\theta\\
&=-2i\Lambda^\theta a'+\Omega^\theta a-\rho^\theta\big(a-\diag(a)\big)\Omega^\theta- \diag(a)\Omega^\theta\\
&=-2i\Lambda^\theta a'+\Omega^\theta \big(a-\diag(a)\big)-\rho^\theta\big(a-\diag(a)\big)\Omega^\theta.
\end{align*}
Clearly, the resulting expression is bounded, and the same notation is kept
for the bounded operator extending it continuously. One also observes that
this operator corresponds to a multiplication operator with compact support.
By adapting the argument of Connes-Moscovici \cite[Prop.~3.2]{CoM} we finally write
\begin{align}
\nonumber \big[F^\theta,a\big]=&\big(1+(D^\theta)^2\big)^{-1/2}D^\theta a
-a\big(1+(D^\theta)^2\big)^{-1/2}D^\theta\\
\label{eqt1}=&\big(1+(D^\theta)^2\big)^{-1/2}\big(D^\theta a-\rho^\theta(a)D^\theta\big)\\
\label{eqt2}&+\big(1+(D^\theta)^2\big)^{-1/2}
\Big(\rho^\theta(a)-\big(1+(D^\theta)^2\big)^{1/2} a\big(1+(D^\theta)^2\big)^{-1/2}\Big)D^\theta.
\end{align}

In order to analyse these terms, let us firstly observe that $(D^\theta)^2$ is a second order uniformly elliptic pseudodifferential operator on $I^\theta\setminus \T^\theta_\epsilon$, and so the inverse of $1+(D^\theta)^2$ is a pseudodifferential operator of order $-2$.
By using the formula
$$
\big(1+(D^\theta)^2\big)^{-1/2}
=\frac{1}{\pi}\int_0^\infty\lambda^{-1/2}\big(1+\lambda+(D^\theta)^2\big)^{-1}\,\d\lambda
$$
one then infers that $\big(1+(D^\theta)^2\big)^{-1/2}$ is a pseudodifferential operator of order $-1$, and accordingly that $\big(1+(D^\theta)^2\big)^{1/2}$ is a pseudodifferential operator of order $1$. 
Note that these properties hold locally on the domain $I^\theta\setminus \T^\theta_\epsilon$.

Now, the term in \eqref{eqt1} is a compact operator since the second factor is a multiplication operator with support in $I^\theta\setminus \T^\theta_\epsilon$ and since the first factor is the pseudodifferential operator of order $-1$ studied above. 
For the term \eqref{eqt2}, we show below that 
$$
\big(1+(D^\theta)^2\big)^{1/2} a\big(1+(D^\theta)^2\big)^{-1/2}=\rho(a)+ R,
$$
where $R$ is a pseudodifferential operator of order $-1$ with compact support. 
The compactness of \eqref{eqt2} follows then directly. 
One finishes the proof with the density of $\AA^\theta_{\epsilon,1}$ in $\AA^\theta$ as $\epsilon \searrow 0$  and the fact that the set of compact operators is norm closed. 

For the last argument, observe that the principal symbols of the operators
$\big(1+(D^\theta)^2\big)^{1/2}$ and $\big(1+(D^\theta)^2\big)^{-1/2}$
are 
$I^\theta\times\R\ni(\lambda,\xi)\mapsto  2\Lambda^\theta(\lambda)|\xi|$ and
$I^\theta\times\R\ni(\lambda,\xi)\mapsto \frac{1}{2}\big(\Lambda^\theta\big)^{-1}(\lambda)|\xi|^{-1}$
respectively. It then follows that
\begin{align*}
(1+(D^\theta)^2)^{1/2} a(1+(D^\theta)^2)^{-1/2}= &\Lambda^\theta \big(a-\diag(a)\big)\big(\Lambda^\theta\big)^{-1}+\diag(a)+R \\
= & \rho^\theta(a) + R
\end{align*}
with $R$ a pseudodifferential operator of order $-1$ with compact support. 
\end{proof}

We are now interested in computing the quotient algebra $\QQ^\theta:=\EE^\theta / \K\big(\Hrond^\theta\big)$. 
Thanks to the previous result, we can focus on elements of the form $\eta(D^\theta)a$
with $\eta \in C_0\big([-\infty,+\infty)\big)$ and $a\in \AA^\theta$.
The starting point is again the decomposition of such elements provided in \eqref{eq:sum}.
As introduced in \eqref{eq:sum} and as explained in the proof of Lemma \ref{lem:compact}, 
it is enough to study the operator  $\eta(D^\theta_{j})\,\!\a_{jj'}$ acting on $\ltwo(I_{j}^\theta)$. If we set $q_{j}:\B\big(\ltwo(I_{j}^\theta)\big)\to \B\big(\ltwo(I_{j}^\theta) \big)/ \K\big(\ltwo(I_{j}^\theta) \big)$ then the image of  $\eta(D^\theta_{j})\,\!\a_{jj'}$ through
$q_{j}$ falls into two distinct situations.
\begin{enumerate}
\item If $j\neq j'$, then 
\begin{equation}\label{eq:q1}
q_{j}\big(\eta(D^\theta_{j})\;\!\a_{jj'}\big) 
=  \eta(-\infty)\;\!\a_{jj'} \in   C_0\big(I^\theta_{j}\cap I^\theta_{j'}\big),
\end{equation}
\item   If $j=j'$, then 
\begin{align}\label{eq:q2}
\begin{split}
q_{j}\big(\eta(D^\theta_{j})\;\!\a_{jj}\big)
& = \Big( \eta\;\!\a_{jj}(\lambda_{j}^\theta-2),\ \eta(-\infty)\;\!\a_{jj},\ 
 \eta\;\!\a_{jj}(\lambda_{j}^\theta+2) \Big) \\ 
& \qquad \in
C_0\big((+\infty,-\infty]\big) \oplus C\big(\overline {I^\theta_{j}}\big) \oplus 
C_0\big([-\infty,+\infty)\big).
\end{split}
\end{align}
\end{enumerate}
These statement can be obtained as in the proof of Lemma \ref{lem:compact} by looking
at these operators in $\ltwo(\R)$ through the conjugation by $\V_{j}^\theta$. Then, one ends up with operators of the form $\eta(D)\varphi(X)$ for $\eta \in C_0\big([-\infty,+\infty)\big)$
and for $\varphi\in C_0(\R)$ in the first case, and $\varphi \in C\big([-\infty,+\infty]\big)$
in the second case. The image of such operators by the quotient map (defined by the compact operators) have been extensively studied in \cite[Sec.~4.4]{Ri}, from which we infer the results presented above. 

\begin{Remark}\label{rem_identification}
Observe that in the first component of \eqref{eq:q2}, the interval $[-\infty,+\infty)$ has been oriented in the reverse direction. The reason is that the function
introduced in \eqref{eq:q2} can be seen as a continuous function on the union of the three intervals
$$
(+\infty,-\infty]\cup \overline {I^\theta_{j}} \cup [-\infty,+\infty)
$$
once their endpoints are correctly identified. This observation and this trick will be used several times in the sequel.
Note also that \eqref{eq:q1} could be expressed as \eqref{eq:q2} by considering
the triple $\big(0, \;\! \eta(-\infty)\a_{jj'},\;\! 0\big)$.
\end{Remark}

In the next statement we collect the various results obtained so far. 
However, in order to provide a unified statement for all $\theta\in (0,\pi)$ and for arbitrary $N$,  some notations have to be slightly updated.
More precisely, recall that for $j\in\{1,\dots,N\}$ one has
$\lambda_j^\theta:=2\cos\left(\frac{\theta+2\pi\;\!j}N\right)$.
For the following statement it will be useful to have a better understanding of the sets $\{\lambda^\theta_j\}_{j=1}^N$. 
Namely, let us observe that for fixed $\theta\in (0,\pi)$ and for $N$ even, we have 
$$
\lambda^\theta_{\frac{N}{2}}<\lambda^\theta_{\frac{N}{2}-1}<\lambda^\theta_{\frac{N}{2}+1}<\lambda^\theta_{\frac{N}{2}-2}<\lambda^\theta_{\frac{N}{2}+2}<\ldots < \lambda^\theta_1<\lambda^\theta_{N-1}<\lambda^\theta_{N},
$$
while for $N$ odd we have
$$
\lambda^\theta_{\frac{N-1}{2}}<\lambda^\theta_{\frac{N+1}{2}}<\lambda^\theta_{\frac{N-1}{2}-1}<\lambda^\theta_{\frac{N+1}{2}+1}<\lambda^\theta_{\frac{N-1}{2}-2}<\ldots<\lambda^\theta_1<\lambda^\theta_{N-1}<\lambda^\theta_{N}.
$$
We now rename these eigenvalues and set $\tl_1<\tl_2<\ldots<\tl_N$ for 
these $N$ distinct  and ordered values.
Accordingly, we define the eigenvector $\tilde \xi^\theta_j$ and the orthogonal projection $\tP_j$, based on the eigenvector $\xi^\theta_k$ and the projection $\P^\theta_k$ corresponding to the eigenvalue $\lambda^\theta_k=\tl_j$.  Finally we set
$$
\tH_j:=\sp\big\{\tP_i\C^N\mid i\leq j\big\}.
$$

\begin{Proposition}\label{prop_quotient}
The $\QQ^\theta:=\EE^\theta / \K\big(\Hrond^\theta\big)$ has the shape
of an upside down comb, with $2N$ teeth, and more precisely:
\begin{equation}\label{eq:comb}
\QQ^\theta \subset C\left(\Big(\bigoplus_{j=1}^{N-1}  \downarrow_j\oplus \rightarrow_j\Big)
\oplus \Big(\downarrow_N \oplus \rightarrow_N \oplus \uparrow^1 \Big) \oplus  \Big( \bigoplus_{j=2}^{N}  \rightarrow^j\oplus \uparrow^j\Big);\B(\C^N)\right)^+,
\end{equation}
with
\begin{align*}
\downarrow_j & := [+\infty,-\infty] \qquad \forall j \in \{1, \ldots,N\}\\
\rightarrow_j & := [\tl_j-2, \tl_{j+1}-2)  \qquad \forall j \in \{1, \ldots,N-1\} \\
\rightarrow_N & :=  [\tl_N-2, \tl_{1}+2], \\
\rightarrow^j & :=  (\tl_{j-1}+2, \tl_{j}+2]  \qquad \forall j \in \{2, \ldots,N\} \\
\uparrow^j & := [-\infty,+\infty]  \qquad \forall j \in \{1, \ldots,N\}.
\end{align*}
Moreover, if $\phi_*$ denotes the restriction to the edge $*$ of any $\phi \in \QQ^\theta$, 
then these restrictions satisfy the conditions:
\begin{align*}
\phi_{\rightarrow_j}  & \in C\big( [\tl_j-2, \tl_{j+1}-2);\B(\tH_j)\big) \ \ \forall j\in \{1, \ldots, N-1\}\\
\phi_{\rightarrow_N}  & \in C\big( [\tl_N-2, \tl_{1}+2]; \B(\C^N)\big) \\
\phi_{\rightarrow^j}  & \in C\big( (\tl_{j-1}+2, \tl_{j}+2]; \B((\tH_{j-1})^\bot)\big) \ \ \forall j\in \{2, \ldots, N\},
\end{align*}
together with
\begin{align*}
&\phi_{\downarrow_j} \in C\big([+\infty,-\infty]; \B(\tP_j \C^N) \big)  \\
&\phi_{\uparrow^{j}} \in C\big([-\infty,+\infty];\B(\tP_j\C^N)\big),
\end{align*}
for all $j\in \{1, \ldots, N\}$.
In addition,  the following continuity properties hold:
\begin{align}\label{eq:cont}
\begin{split}
\phi_{\rightarrow_1}(\tl_1-2) &= \phi_{\downarrow_1}(-\infty), \\
\phi_{\rightarrow_j}(\tl_j-2) &= \lim_{\lambda \nearrow \tl_j-2}\phi_{\rightarrow_{j-1}}(\lambda)
\oplus \phi_{\downarrow_j}(-\infty) \ \ \forall j\in \{2, \ldots, N\} \\
\phi_{\rightarrow_N}(\tl_1+2)  &=   \phi_{\uparrow^{1}}(-\infty) \oplus \lim_{\lambda \searrow \tl_1+2}\phi_{\rightarrow^{2}}(\lambda) \\
\phi_{\rightarrow_j}(\tl_j+2)  &=  \phi_{\uparrow^{j}}(-\infty) \oplus \lim_{\lambda \searrow \tl_j+2}\phi_{\rightarrow^{j+1}}(\lambda) \ \ \forall j\in \{2, \ldots, N-1\} \\
\phi_{\rightarrow^N}(\tl_N+2)  &= \phi_{\uparrow^{N}}(-\infty),
\end{split}
\end{align}
and there exists $c\in \C$ such that for all $j\in \{1, \dots,N\}$,
\begin{equation}\label{eq:unit}
c=\Tr\big(\phi_j(+\infty)\big)=\Tr\big(\phi^j(+\infty)\big).
\end{equation}
\end{Proposition}

A representation for the support of the quotient algebra $\QQ^\theta$ is provided in Figure \ref{fig_Rt}. In the previous description of the quotient algebra, note that the condition \eqref{eq:unit} is related to the addition of the unit to $\EE^\theta$. Indeed, if no unit is added to $\EE^\theta$, then one has $c=0$. 
Let us also denote by $q^\theta$ the quotient map 
$$
q^\theta: \EE^\theta \to \QQ^\theta\equiv \EE^\theta/\K\big(\Hrond^\theta\big).
$$

\begin{figure}
    \centering
    \includegraphics[width=15cm]{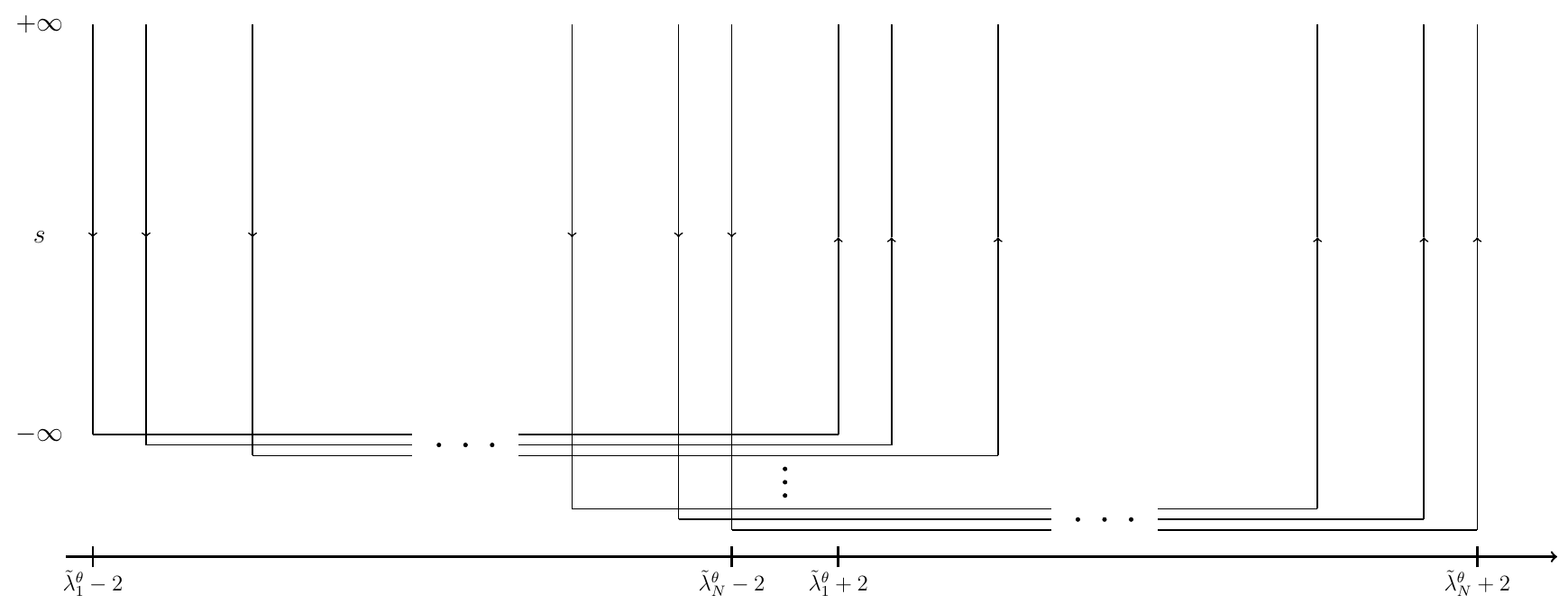}
    \caption{A representation of the quotient algebra $\QQ^\theta$.}
    \label{fig_Rt}
\end{figure}

\begin{proof}
The proof consists in looking carefully at the expressions provided in \eqref{eq:q1} and \eqref{eq:q2}, and in keeping track of the matricial form of the matrix-valued function $a$.
Let us consider a generic element of $\EE^\theta$ given by
$\eta(D^\theta)a + c1$ for $\eta\in C_0\big([-\infty,+\infty)\big)$, $a\in \AA^\theta$,
and $c\in \C$. 
By taking \eqref{eq:sum} into account, and the specific form of $\Hrond^\theta$,
one has
\begin{equation*}
\eta(D^\theta)a + c1
=  \sum_{j,j'}
\big(\eta(\tilde D^\theta_{j})\;\!\tilde\a_{jj'}+c \delta_{jj'}\big)\otimes 
\;\!\tfrac{1}{N}|\tilde\xi^\theta_j\rangle \langle \tilde\xi^\theta_{j'}|
\end{equation*}
with $\delta_{jj'}$ the Kronecker delta function.
According to \eqref{eq:q1} and \eqref{eq:q2} one has
\begin{align*}
&q^\theta\big(\eta(D^\theta)a + c1\big) \\
& = \sum_{\{j,j'\mid j'\neq j\}}
\eta(-\infty)\;\!\tilde \a_{jj'}\otimes 
\;\!\tfrac{1}{N}|\tilde \xi^\theta_j\rangle \langle \tilde \xi^\theta_{j'}| \\
&\quad +
\sum_{j}
\Big( \eta\;\!\tilde \a_{jj}(\tilde \lambda_{j}^\theta-2)+c,\ \eta(-\infty)\;\!\tilde \a_{jj}+c, 
\ \eta\;\!\tilde \a_{jj}(\tilde \lambda_{j}^\theta+2)+c \Big)\otimes 
\;\!\tfrac{1}{N}|\tilde \xi^\theta_j\rangle \langle \tilde \xi^\theta_{j}| \\
& = \sum_{j,j'}
\Big( \big(\eta\;\!\tilde \a_{jj}(\tilde \lambda_{j}^\theta-2)+c\big)\delta_{jj'},\  
\eta(-\infty)\;\!\tilde \a_{jj'}+c\delta_{jj'}, \
\big(\eta\;\!\tilde \a_{jj}\tilde \lambda_{j}^\theta+2)+c\big)\delta_{jj'} \Big)\otimes 
\;\!\tfrac{1}{N}|\tilde \xi^\theta_j\rangle \langle \tilde \xi^\theta_{j'}|.
\end{align*}

In order to fully understand the previous expression, it is necessary to remember
the special structure of the underlying Hilbert space 
$\Hrond^\theta:=\int_{I^\theta}^\oplus\Hrond^\theta(\lambda)\;\!\d\lambda$
with 
\begin{align*}
\Hrond^\theta(\lambda)
& =
\begin{cases}  \tilde \Hrond^\theta_j  & \hbox{ if } \tilde\lambda_j^\theta-2\leq \lambda< \tilde\lambda_{j+1}^\theta-2 \\
\C^N & \hbox{ if } \tilde \lambda^\theta_N-2 \leq \lambda \leq \tilde \lambda^\theta_1+2\\
(\tilde \Hrond^\theta_j)^\bot  & \hbox{ if } \tilde\lambda_j^\theta+2<\lambda\leq \tilde\lambda_{j+1}^\theta+2 
\end{cases}. 
\end{align*}
Note that compared with the original definition of $\Hrond^\theta(\lambda)$ we have changed the fiber at a finite number of points, which does 
not impact the direct integral, but simplify our argument subsequently.
Thus, the changes of dimension of the fibers take place at all $\tilde \lambda_j^\theta-2$ 
and $\tilde \lambda_j^\theta+2$, for $j\in \{1, \dots,N\}$.
By taking this into account, the interval $I^\theta$ has to be divided into $2N-1$ subintervals, 
firstly of the form $[\tilde \lambda_j^\theta-2, \tilde \lambda_{j+1}^\theta-2)$ for
$j\in \{1,\dots,N-1\}$, then
the special interval $[\tilde \lambda_N^\theta-2,\tilde \lambda_1^\theta+2]$, and finally
the intervals $(\tilde \lambda_{j-1}^\theta+2,\tilde \lambda_{j}^\theta+2]$ for $j\in \{2, \dots, N\}$.
By using this partition of $I^\theta$, one gets
\begin{align*}
&q^\theta\big(\eta(D^\theta)a + c1\big) \\
&=\sum_{j=1}^{N-1}
\Big(\big(\eta\;\!\tilde \a_{jj}(\tilde \lambda_{j}^\theta-2)+c\big)\otimes
\tilde \P_j^\theta,\  \chi_{[\tilde \lambda_j^\theta-2, \tilde \lambda_{j+1}^\theta-2)}\big(
\eta(-\infty)\;\!a+c1_{\tilde{\Hrond}^\theta_j}\big)\Big) \\
& \quad +
\Big(\big(\eta\;\!\tilde \a_{NN}(\tilde \lambda_{N}^\theta-2)+c\big)\otimes\tilde \P_N^\theta,
\  \chi_{[\tilde \lambda_N^\theta-2,\tilde \lambda_1^\theta+2]}\big(\eta(-\infty)\;\!a+c1\big),\  
\big(\eta\;\!\tilde \a_{11}(\tilde \lambda_{1}^\theta+2)+c\big)\otimes \tilde \P_1^\theta \Big) \\
& \quad + \sum_{j=2}^{N} 
\Big(\chi_{(\tilde \lambda_{j-1}^\theta+2, \tilde \lambda_{j}^\theta+2]}\big(\eta(-\infty)\;\!a+c1_{(\tilde{\Hrond}^\theta_{j-1})^\bot}\big), \ 
\big(\eta\;\!\tilde \a_{jj}(\tilde \lambda_{j}^\theta+2)+c\big)\otimes \tilde \P_{j}^\theta \Big).
\end{align*}
The description obtained so far leads directly to the structure of \eqref{eq:comb}.

The properties stated in \eqref{eq:cont} follow from the continuity of $a\in \AA^\theta$
and from its properties at thresholds. 
The final property \eqref{eq:unit} comes from the unit and the fact that 
$\eta$ vanishes at $+\infty$.
\end{proof}

Since $\F^\theta W_-^\theta(\F^\theta)^*\in \EE^\theta$, by Proposition \ref{prop:affiliation},
one can look at the image of this operator in the quotient algebra.
The following statement contains a description of this image, using the notations introduced
in Proposition \ref{prop_quotient}. 
The functions $\eta_\pm:\R\to \C$ defined for any $s\in \R$ by
\begin{equation}\label{eq_def_eta}
\eta_\pm(s):=\tanh(\pi s)\pm i \cosh(\pi s)^{-1}.
\end{equation}
will also be used.
Finally, based on the projections $\{\tilde \P^\theta_j\}_{j=1}^N$
introduced before Proposition \ref{prop_quotient}, we define the channel
scattering matrix $\tilde S^\theta_{jj}(\lambda):=\tilde \P_j^\theta S^\theta(\lambda)\tilde \P_j^\theta$. 

\begin{Lemma}\label{lem:restrictions}
Let $\phi:=q^\theta \big(\F^\theta W_-^\theta(\F^\theta)^*\big)$
denote the image of $\F^\theta W_-^\theta(\F^\theta)^*$ in the quotient algebra. 
Then, the restrictions of $\phi$ on the various parts of $\QQ^\theta$ are given by:
\begin{align*}
& \hbox{ for }  j\in \{1, \ldots, N-1\} \hbox{ and } \lambda \in [\tl_j-2, \tl_{j+1}-2), \quad & \phi_{\rightarrow_j}(\lambda) = S^\theta(\lambda) \ \in \B(\tH_j)\\
&\hbox{ for } \lambda \in  [\tl_N-2, \tl_{1}+2], \quad & \phi_{\rightarrow_N}(\lambda) 
= S^\theta(\lambda)\  \in \B(\C^N)\\
&\hbox{ for }  j\in \{2, \ldots, N\} \hbox{ and }  \lambda \in  (\tl_{j-1}+2, \tl_{j}+2], \quad & \phi_{\rightarrow^j}(\lambda) = S^\theta(\lambda) \ \in  \B\big((\tH_{j-1})^\bot\big),
\end{align*}
and for $s\in \R$
\begin{align}
\label{eq:res1} &\phi_{\downarrow_j}(s)= 1+ \tfrac12\big(1-\eta_-(s)\big)\big(\tilde S^\theta_{jj}(\tilde\lambda^\theta_j-2)-1\big) \ \in \B(\tP_j \C^N) \\
\label{eq:res2} &\phi_{\uparrow^{j}}(s)= 1+ \tfrac12\big(1-\eta_+(s)\big)\big(\tilde S^\theta_{jj}(\tilde\lambda^\theta_j+2)-1\big) \ \in\B(\tP_j\C^N).
\end{align}
\end{Lemma}

\begin{proof}
Let us start by looking at the expression for
$\F^\theta \big(W_-^\theta-1\big)(\F^\theta)^*$, as provided for example in 
\eqref{eq_W_Hrond}, and by taking the content of Lemma \ref{lem_tanh}
into account. It then follows that $\F^\theta \big(W_-^\theta-1\big)(\F^\theta)^*$
can be rewritten as
\begin{equation}\label{eq:2terms}
\tfrac12\big(1-\tanh(\pi D^\theta)\big)\big(S^\theta(X^\theta)-1\big)
-i\tfrac{1}{2}\cosh(\pi D^\theta)^{-1} \sum_j
\tfrac{X^\theta-\lambda^\theta_j}{2}\P^\theta_j 
\big(S^\theta(X^\theta)-1\big)
+k^\theta
\end{equation}
with $k^\theta\in \K\big(\Hrond^\theta\big)$.
Thus, one ends up with two generic elements of $\EE^\theta$ which have been carefully
studied in the proof of Proposition \ref{prop_quotient}. The only additional necessary tricky observation is that
$$
\tfrac{\lambda-\lambda^\theta_j}{2}\Big|_{\lambda = \lambda_j^\theta\pm 2} = \pm 1
$$
which explains the appearance of the two functions $\eta_\pm$.
Note also that 
$$
\lim_{s\to -\infty}\tfrac12\big(1-\tanh(\pi s)\big) = 1 \quad \hbox{ and }\quad
\lim_{s\to \infty}\tfrac12\big(1-\tanh(\pi s)\big) = 0, 
$$
while $\lim_{s\to \pm \infty}\cosh(\pi s)^{-1} = 0$.
The rest of the proof is just a special instance of the proof of Proposition \ref{prop_quotient}
for the two main terms exhibited in \eqref{eq:2terms}, since the compact term verifies
$q^\theta(k^\theta)=0$.
\end{proof}

\section{\texorpdfstring{$K$}{K}-theory for the quotient algebra}\label{sec:Kth}
\setcounter{equation}{0}

In this section, we compute the $K$-groups of the quotient algebra. Since this
computation is of independent interest and does not rely on the details of the model,
we provide a self-contained proof with simpler notations (the longest proofs are given in the Appendix). 
The main difficulty in the quotient algebra is the appearance of continuous functions with values in matrices of different sizes. However, the continuity conditions are very strict, and allow us to 
determine the $K$-theory for this algebra.

Proposition \ref{prop_quotient} tells us that the quotient algebra is matrix-valued functions on the union of the boundaries of the squares $I_j^\theta\cap I_{j'}^\theta\times[-\infty,\infty]$. Proposition \ref{prop_quotient} also describes how the rank of the matrices increases with $\lambda$ until we reach $\tilde\lambda^\theta_1+2$ when it decreases again. The changes of rank happen continuously as the ``extra rank'' appears from (or disappears to) infinity via the ``vertical'' functions $\phi_{\downarrow_j}$ and $\phi_{\uparrow^j}$. 
This section repeats this construction in a more induction friendly way, so that the $K$-theory can be computed. Interval by interval the isomorphism of the two constructions will be clear, and that the gluing of the pieces is the same can be checked using Proposition \ref{prop_quotient}. 

Before starting the construction, let us recall a result which will be constantly used, see
\cite[Ex.~4.10.22]{HR}.
Consider three $C^*$-algebras $A, B$, and $C$, two surjective $*$-homomorphisms 
$\pi_1: A\to C$ 
and $\pi_2: B\to C$, and  the pullback diagram
\[
\xymatrix{ A\oplus_CB\ar[r]^{}\ar[d]^{} & A\ar[d]^{\pi_1}\\
B\ar[r]_{\pi_2} & C}
\]
with the pullback algebra $A\oplus_CB=\big\{(a,b)\in A\oplus B\mid \pi_1(a)=\pi_2(b)\big\}$.
Then, the Mayer-Vietoris sequence 
\begin{equation}\label{eq_MV}
\begin{split}
\xymatrix{K_0(A\oplus_CB)\ar[r]^{} & K_0(A)\oplus K_0(B)\ar[r]^-{\pi_{1*}-\pi_{2*}} & K_0(C)\ar[d]\\
K_1(C)\ar[u] & K_1(A)\oplus K_1(B)\ar[l]^-{\pi_{1*}-\pi_{2*}} & K_1(A\oplus_CB)\ar[l]^{}
}
\end{split}
\end{equation}
allows us to compute the $K$-theory of the pullback algebra in terms of the building blocks.
The map $\pi_{1*}-\pi_{2*}$ is often referred to as the difference homomorphism.

Let us also introduce two easy lemmas, whose proofs are elementary exercises in $K$-theory.
For shortness, we shall use the notation $I$ for $(0,1]$, the half-open interval.

\begin{Lemma}\label{lem:K1}
One has $K_0\big(C_0(I)\big)=K_1\big(C_0(I)\big)=0$.
\end{Lemma}

For the next statement, we introduce the algebra $D_j$ for $j\geq 2$ by
$$
D_j:=\big\{f:[0,1]\to \B(\C^j)\mid f(0)\in \B(\C^{j-1})\oplus\C\big\}.
$$
In other words, $D_j$ is made of function on $[0,1]$ with values in $\B(\C^j)$ with the only condition that $f(0)$ is block diagonal, with one block of size $(j-1)\times (j-1)$ and one block
of size $1\times 1$. Note that we identity $\C$ with $\B(\C)$.

\begin{Lemma}\label{lem:K2}
The short exact sequence
$$
0\to C_0\big((0,1)\big)\ox \B(\C^j)\to D_j\to \B(\C^{j-1})\oplus\C\oplus \B(\C^j)\to 0
$$
gives $K_0(D_j)=\Z^{2}$, $K_1(D_j)=0$.
\end{Lemma}

Let us now construct inductively some algebras. Again, $I$ denotes $(0,1]$, 
and let $\{P_j\}_{j=1}^N$ be the rank one projections based on the standard basis of 
$\C^N$.
For $j\in \{1, \dots,N\}$ set $A_j:=C_0\big(I;\B(P_j\C^N)\big)$, and let us also 
fix $A(1):=A_1$. For $j\geq 2$ let $A(j)$ be the algebra obtained by gluing
$A(j-1)\oplus A_j$ to $D_j$, namely 
\[
A(j)=\{(f\oplus g,h)\in A(j-1)\oplus A_j \oplus D_j:\,f(1)\oplus g(1) = h(0)\}.
\]
Note that once the algebras $A(j-1)\oplus A_j$ and $D_j$ are glued
together, the resulting algebra $A(j)$ is considered again as matrix-valued functions
defined on $[0,1]$. It means that a reparametrization of the interval is performed at 
each step of the iterative process. 
Note also that the algebra $A(j)$ can be thought as an algebra of functions with values in a set of matrices
of increasing size. However, the change of dimension from $k-1$ to $k$ is performed concomitantly with the addition of the new path $A_k$ linking $0$ to the new diagonal entry of the $k\times k$ matrix.   

\begin{Proposition}\label{prop:step1}
For any $j\in \{1,\dots,N\}$ one has 
$K_0\big(A(j)\big)=0$ and $K_1\big(A(j)\big)=0$.
\end{Proposition}

The proof of this statement is provided in the Appendix.
Note that the same result holds for the matrices $B(j)$ obtained by reversing the construction
(and by taking care of the subspaces of $\C^N$ involved). For this reverse construction, 
we firstly set for $j\in \{1,\dots,N\}$
$$
E_j:=\big\{f:[0,1]\to \B\big((\C^{N-j})^\bot\big)\mid f(1)\in \C\oplus \B\big((\C^{N-j+1})^\bot\big)\big\}
$$
and $B_j:=C_0\big([0,1);\B(P_{N-j+1}\C^N)\big)$.
We also fix $B(1):= B_1$.
For $j\geq 2$ let $B(j)$ be the algebra obtained by gluing
$E_j$ to $B_j\oplus B(j-1)$, namely for  $h\in E_j$ and for $f\oplus g\in B_j\oplus B(j-1)$
we set $h(1)= f(0)\oplus g(0)$.
Note that once the algebras $E_j$ and $B_j\oplus B(j-1)$ are glued
together, the resulting algebra $B(j)$ is considered again as matrix-valued functions
defined on $[0,1]$. It means that a reparametrization of the basis is performed at 
each step of the iterative process. 
The algebra $B(j)$ can be thought as an algebra of functions with values in a set of matrices
of decreasing size. However, the change of dimension from $k$ to $k-1$ is performed concomitantly with the addition of the new path $B_k$ linking the removed entry to $0$.

By exactly the same proof, we get the same result:

\begin{Proposition}\label{prop:step1bis}
For any $j\in \{1,\dots,N\}$ one has 
$K_0\big(B(j)\big)=0$ and $K_1\big(B(j)\big)=0$.
\end{Proposition}

The next step will be to glue together an algebra of the family $A(j)$ with an algebra of the family $B(j)$. Before this, we need another lemma, whose proof is again an elementary exercise.

\begin{Lemma}\label{lem:K3}
For 
$$
C_N:=\big\{f:[0,1]\to \B(\C^N) \mid f(0)\in \B(\C^{N-1})\oplus\C,\ f(1)\in \C\oplus \B\big((\C^1)^\bot\big)\big\},
$$
the exact sequence
$$
0\to C_0\big((0,1)\big)\ox \B(\C^N)\to C_N\to \B(\C^{N-1})\oplus\C\oplus \C\oplus \B\big((\C^1)^\bot\big)\to 0
$$
gives $K_0(C_N)=\Z^{3}$ and $K_1(C_N)=0$.
\end{Lemma}

We can now collect the information obtained so far, and provide the $K$-theory
for the main algebra of this section. The proof of the statement is provided in 
the Appendix.

\begin{Proposition}\label{prop:step2}
For 
$$
Q_N=\big(A(N-1)\oplus A_N\big)\oplus_{\B(\C^{N-1})\oplus\C}C_N
\oplus_{\C\oplus \B((\C^1)^\bot)}\big(B_N\oplus B(N-1)\big)
$$
one has
$K_0\big(Q_N^+\big)=\Z$ and $K_1\big(Q_N^+\big)=\Z$.
\end{Proposition}

 \begin{figure}[h]
     \centering
     \includegraphics[width=15cm]{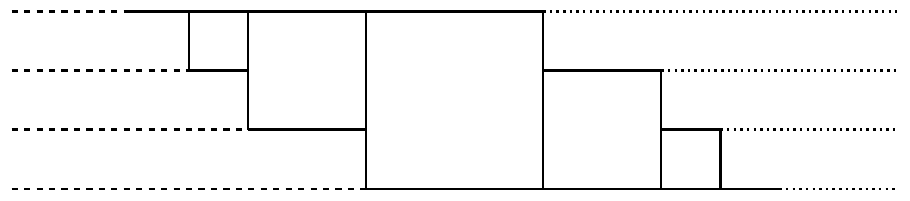}
     \caption{Representation of the algebra $Q_4$, with the biggest square representing $C_4$, while the part of the left of this square corresponds to $A(3)\oplus A_4$, and the part on the right corresponds to $B_4\oplus B(3)$. For comparison with the algebra $\QQ^\theta$, the dashed lines represent the support of the functions $\phi_{\downarrow_j}$ while the dotted lines represent the support of the functions $\phi_{\uparrow^j}$. The remaining lines and squares correspond to the support of the scattering matrix.}
     \label{fig_Adam}
 \end{figure}

By construction, the algebras $\QQ^\theta$ and $Q_N^+$ are isomorphic, therefore they share
the same $K$-theory. Note that the isomorphism can also be inferred by comparing Figures 
\ref{fig_Rt} and \ref{fig_Adam}\;\!: the vertical lines of Figure \ref{fig_Rt} are symbolised horizontally in Figure \ref{fig_Adam}, with dashed or dotted lines. Furthermore, since these lines represent the support of a block diagonal matrix-valued function, they can be rescaled independently.
Thanks to this isomorphism, the $K$-theory of $\QQ^\theta$ can be
directly deduced.

\begin{Corollary}
One has $K_0\big(\QQ^\theta\big)=\Z$ and $K_1\big(\QQ^\theta\big)=\Z$.
\end{Corollary}

\section{Topological Levinson's theorem}\label{sec:Lev}
\setcounter{equation}{0}

In this section we provide the topolgical version of Levinson's theorem, linking
the number of bound states to an expression involving the image of $W^\theta_-$
in the quotient algebra. First of all, we provide a statement about the behavior
of the scattering matrix at thresholds. It shows that the expression obtained for 
the wave operator $W_-^\theta$ is rather rigid and imposes a strict behaviour at thresholds.
For its statement we define the scalar valued function $\tilde \s_{jj}^\theta$ by the relation 
$\tilde S^\theta(\lambda)_{jj}=:\tilde\s_{jj}^\theta(\lambda)\tilde\P_j^\theta$. 

\begin{Lemma}\label{lem:11}
For any $j\in \{1,\dots,N\}$ one has $\tilde\s^\theta_{jj}(\tilde\lambda^\theta_j\pm 2)\in \{-1,1\}$.
\end{Lemma}

\begin{proof}
Since $W_-^\theta$ is a Fredholm operator of norm 1, the image
$q^\theta \big(\F^\theta W_-^\theta(\F^\theta)^*\big)$ is a unitary operator, and therefore
its restrictions on all components of $\QQ^\theta$ must be unitary.
The restrictions on $\rightarrow_j$, $\rightarrow_N$, and $\rightarrow^j$ do not impose
any conditions, since the scattering operator is unitary valued, see  Lemma \ref{lem:restrictions}.
On the other hand, by checking that the restrictions on $\downarrow_j$ and 
on $\uparrow^j$, the conditions appearing in the statement have to be imposed. 

Starting with \eqref{eq:res1}, the restriction on $\downarrow_j$, we can rewrite this operator
as a scalar function multiplying a rank one projection. By imposing that the operator is unitary valued, one infers that
$$
\Big(1+ \tfrac12\big(1-\eta_-(s)\big)\big({\tilde{\s}}^\theta_{jj}(\tilde\lambda^\theta_j-2)-1\big)\Big)
\Big(1+ \tfrac12\big(1-\eta_-(s)\big)\big(\tilde{\s}^\theta_{jj}(\tilde\lambda^\theta_j-2)-1\big)\Big)^*=1
$$
for all $s\in \R$. Some direct computations lead then to the condition
$\Im\big(\eta_-(s)\big)\Im\big({\tilde{\s}}^\theta_{jj}(\tilde\lambda^\theta_j-2)\big)=0$
for any $s\in \R$, meaning that ${\tilde{\s}}^\theta_{jj}(\tilde\lambda^\theta_j-2)\in \R$.
Since ${\tilde{\s}}^\theta_{jj}(\tilde\lambda^\theta_j-2)$ is also unitary valued, 
the only solutions are the ones given in the statement. A similar argument holds for 
${\tilde{\s}}^\theta_{jj}(\tilde\lambda^\theta_j+2)$, starting with the restriction on 
$\uparrow^j$.
\end{proof}

The topological Levinson's theorem corresponds to an index theorem in scattering theory.
By considering the $C^*$-algebras introduced in Section \ref{sec:C*} we can consider the short
exact sequence of $C^*$-algebras
$$
0 \longrightarrow \K\big(\Hrond^\theta\big)\longrightarrow \EE^\theta \stackrel{q^\theta}{\longrightarrow}  \QQ^\theta \longrightarrow 0.
$$
Since $\F^\theta W_-^\theta(\F^\theta)^*$ belongs to $\EE^\theta$ and is a Fredholm operator
we infer the equality
\begin{equation}\label{eq:Lev1}
\ind\Big([q^\theta \big(\F^\theta W_-^\theta(\F^\theta)^*\big)]_1\Big)  
=-\big[E_{\rm p}(H^\theta)\big]_0,
\end{equation}
where $\ind$ denotes the index map from $K_1\big(\QQ^\theta\big)$ to $K_0\big( \K\big(\Hrond^\theta\big)\big)$ and where $E_{\rm p}(H^\theta)$ corresponds to the projection on the subspace spanned by the eigenfunctions of $H^\theta$.
Note that this projection appears from the standard relation 
$$
\big[1-(W^\theta_-)^* W_-^\theta\big]_0-\big[1-W_-^\theta (W_-^\theta)^*\big]_0 
=-\big[E_{\rm p}(H^\theta)\big]_0.
$$

The equality \eqref{eq:Lev1} can be directly deduced from \cite[Prop.~4.3]{Ri}.
Let us emphasise that this equality corresponds to the topological version of Levinson's theorem: it is a relation (by the index map) between the equivalence class in $K_1$
of quantities related to scattering theory, as described in Lemma \ref{lem:restrictions}, and the equivalence class in $K_0$ of the projection on the bound states of $H$. 
However, the standard formulation of Levinson's theorem is an equality between numbers. 
Thus, our final task is to extract a numerical equality from \eqref{eq:Lev1}.

In the next statement, the notation 
$\Var \big(\lambda \mapsto \det S^\theta(\lambda)\big)$ 
should be understood as  the total variation of the argument of the piecewise continuous function 
\begin{equation*}
I^\theta \ni \lambda \mapsto \det S^\theta(\lambda) \in \S^1.
\end{equation*}
where we compute the argument increasing with increasing $\lambda$. Our convention is also that the increase of the argument is counted clockwise.

\begin{Theorem}\label{thm:Adam}
For any $\theta \in (0,\pi)$ the following equality holds:
\begin{equation}\label{eq:ouf}
\# \sigma_{\rm p}(H^\theta)=N-\frac{\#\big\{j\mid \s_{jj}(\lambda^\theta_j\pm 2)=1\big\}}2+ \Var \big(\lambda \to \det S^\theta(\lambda)\big).
\end{equation}
\end{Theorem}

\begin{proof}
The proof starts by evaluating both sides of \eqref{eq:Lev1} with 
the operator trace to obtain a numerical equation.
For the right hand side, we obtain (minus) the number of bound states, or more precisely
$\Tr\big(E_{\rm p}(H^\theta)\big)= \# \sigma_{\rm p}(H^\theta)$ if the multiplicity of the eigenvalues is taken into account.

For the left hand side, recall that $\QQ^\theta$ is isomorphic to $Q_N^+$, with $Q_N$ introduced in 
Proposition \ref{prop:step2}, and that this algebra is naturally embedded (after rescaling) in $C_0\big(\R;\B(\C^N)\big)^+$. Thus, the index of $W_-$ is computed by the winding number of the pointwise determinant of $q(W_-)$, see for example \cite[Prop.~7]{KR08wind}. The various contributions for this computation can be inferred either from Figure
\ref{fig_Rt} or from Figure \ref{fig_Adam}, but the functions to be considered are provided by
Lemma \ref{lem:restrictions}. If we use the representation of the quotient algebra provided in 
Figure \ref{fig_Rt}, then all horizontal contributions $\phi_{\to j}$ can be encoded in the expression 
$\Var \big(\lambda \to \det S^\theta(\lambda)\big)$. For $\lambda\mapsto S^\theta(\lambda)$ piecewise $C^1$ these contributions can be computed analytically by the formula
\begin{align*}
&- \frac{1}{2\pi i} \sum_{k=1}^{N-1}\int_{\tl_k-2}^{\tl_{k+1}-2}\Tr\big((S^\theta)^*(S^\theta) '\big)(\lambda)\,\d\lambda
- \frac{1}{2\pi i} \int_{\tl_N-2}^{\tl_{1}+2}\Tr\big((S^\theta)^*(S^\theta) '\big)(\lambda)\,\d\lambda \\
&\quad - \frac{1}{2\pi i} \sum_{k=1}^{N-1}\int_{\tl_k+2}^{\tl_{k+1}+2}\Tr\big((S^\theta)^*(S^\theta) '\big)(\lambda)\,\d\lambda
\end{align*}
where $S^\theta$ is $C^1$ on each interval.

For the vertical contributions, recall that the functions $\eta_\pm$ have been introduced in \eqref{eq_def_eta}.
These contributions have to be computed from $+\infty$
to $-\infty$ for the intervals with one endpoint at $\tl_j-2$, while they have to be computed from $-\infty$ to $+\infty$
for the intervals having an endpoint at $\tl_j+2$. 

In both cases, if $\tilde S^\theta_{jj}(\tilde\lambda^\theta_j\pm 2)=1$, then the contribution is $0$, as a straightforward consequence of \eqref{eq:res1} and
\eqref{eq:res2}. On the other hand, if  $\tilde S^\theta_{jj}(\tilde\lambda^\theta_j- 2)=-1$,
then $\phi_{\downarrow_j}= \eta_-$, and this leads to a contribution of $\frac{1}{2}$, with our clockwise convention for the increase of the variation.
Similarly, if  $\tilde S^\theta_{jj}(\tilde\lambda^\theta_j+ 2)=-1$, then 
$\phi_{\uparrow^{j}}= \eta_+$ and this leads again to a contribution of $\frac{1}{2}$,
because of the change of orientation of the path. As a consequence, if
$\tilde S^\theta_{jj}(\tilde\lambda^\theta_j\pm 2)=-1$, the corresponding contribution is $\frac{1}{2}$, no matter if it is the opening or the closing of a channel of scattering.
Our presentation in \eqref{eq:ouf}, which takes the content of Lemma \ref{lem:11} into account,  reflects the genericity of the value $-1$ at thresholds over the value $1$.
\end{proof}

\section{Explicit computations for \texorpdfstring{$N=2$}{N=2}}\label{sec:N=2}
\setcounter{equation}{0}

In this section, we concentrate on the case $N=2$ and make the computations
as explicit as possible.
We firstly recall some notations restricted to the case $N=2$ and $\theta\in (0,\pi)$. For the analysis of $H^\theta$, the main spectral result is a necessary and sufficient condition for the existence of an eigenvalue for the operator $H^\theta$.
For that purpose, we decompose the matrix $\diag(v):=\diag\big(v(1),v(2)\big)$
as a product $\diag(v)=\u\;\!\v^2$,
where $\v:=|\diag(v)|^{1/2}$ and $\u$ is the diagonal matrix
with components
$$
\u_{jj}=
\begin{cases}
+1 &\hbox{if $v(j)\ge0$}\\
-1 &\hbox{if $v(j)<0$,}
\end{cases}
$$
for $j\in\{1,2\}$.
For simplicity, we set  $v(1)=u_1 a^2$, $v(2)=u_2b^2$ and assume that $a\geq b\geq 0$ 
and $a>0$. We also impose the following condition: if $u_1=u_2$, 
then $a>b$, since otherwise the function $v$ would be $1$--periodic and not $2$--periodic.
With these notations one has
$\v=\left(\begin{smallmatrix}a & 0 \\ 0 & b
\end{smallmatrix}\right)$ and $
\u=\left(\begin{smallmatrix}
u_1 & 0 \\ 0 & u_2
\end{smallmatrix}\right)$.
Let us also set
\begin{equation*}
\varrho:=u_2a^2+u_1b^2
\end{equation*}
since this expression will often appear in the sequel.

For $N=2$, obseve that  
$\lambda_1^\theta=-2\cos(\tfrac{\theta}{2})$ and $\lambda_2^\theta=2\cos(\tfrac{\theta}{2})$, $\xi_1^\theta=\left(\begin{smallmatrix}-\e^{i\theta/2} \\ \e^{i\theta}
\end{smallmatrix}\right)$  and $ \xi_2^\theta=\left(\begin{smallmatrix}\e^{i\theta/2} \\ \e^{i\theta}\end{smallmatrix}\right)$.
We then set 
\begin{equation*}
\Xi_+ (\theta):=\big(\cos (\tfrac{\theta}{2})+ \cos^2 (\tfrac{\theta}{2})\big)^{\frac12}>0, \quad  \Xi_- (\theta):=\big(\cos (\tfrac{\theta}{2})-\cos^2 (\tfrac{\theta}{2})\big)^{\frac12}>0.
\end{equation*}
For later use, it is useful to observe that $\Xi_+$ is a strictly decreasing function on $(0,\pi)$
with range in $(0,\sqrt{2})$. Finally, for $z\in\C$ and $ j\in\{1,2\}$ we introduce the expression
\begin{equation*}
\beta_j^\theta(z):=\big|\big(z-\lambda_j^\theta\big)^2-4\big|^{1/4}.
\end{equation*}

\subsection{Values of \texorpdfstring{$S^\theta(\lambda)$}{S(lambda)} at thresholds}\label{subsec:thresholds}

In this section, we provide information about the scattering matrix at thresholds.
For conciseness, we restrict our attention to  $\lambda\in (-4,0)$.
In this setting, only two thresholds appear: one at $\lambda_1^\theta-2$ and 
one at $\lambda_2^\theta-2$. The first one corresponds to the threshold at the bottom of the essential spectrum, while the second one corresponds to an embedded threshold. Since we we already know from \cite[Thm.~3.9]{NRT} that $S^\theta_{11}$ is continuous at $\lambda_2^\theta-2$ we postpone its study to Section \ref{subsec:inside}. 

\begin{Proposition}\label{S(lambda)N=2}
Let $\theta \in(0,\pi)$ and consider $a\geq b \geq 0$ with $a>0$.
Then the following equalities hold:
\begin{align}
\label{eq:k1} S^\theta(\lambda_1^\theta-2)_{11} & =
\P_1^{\theta} \quad \hbox{ if }\  2\varrho\Xi_+(\theta)+a^2b^2=0
\quad \hbox{(resonant case)}, \\
\label{eq:k3} S^\theta(\lambda_1^\theta-2)_{11} & =
 -\P_1^{\theta} \quad \hbox{ otherwise } \qquad \hbox{(generic case)}, 
\end{align}
and 
\begin{equation*}   
S^\theta(\lambda_2^\theta-2)_{22}=-\P_2^\theta.  
\end{equation*}
\end{Proposition}

\begin{figure}
\centering
  \includegraphics[width=.7\textwidth]{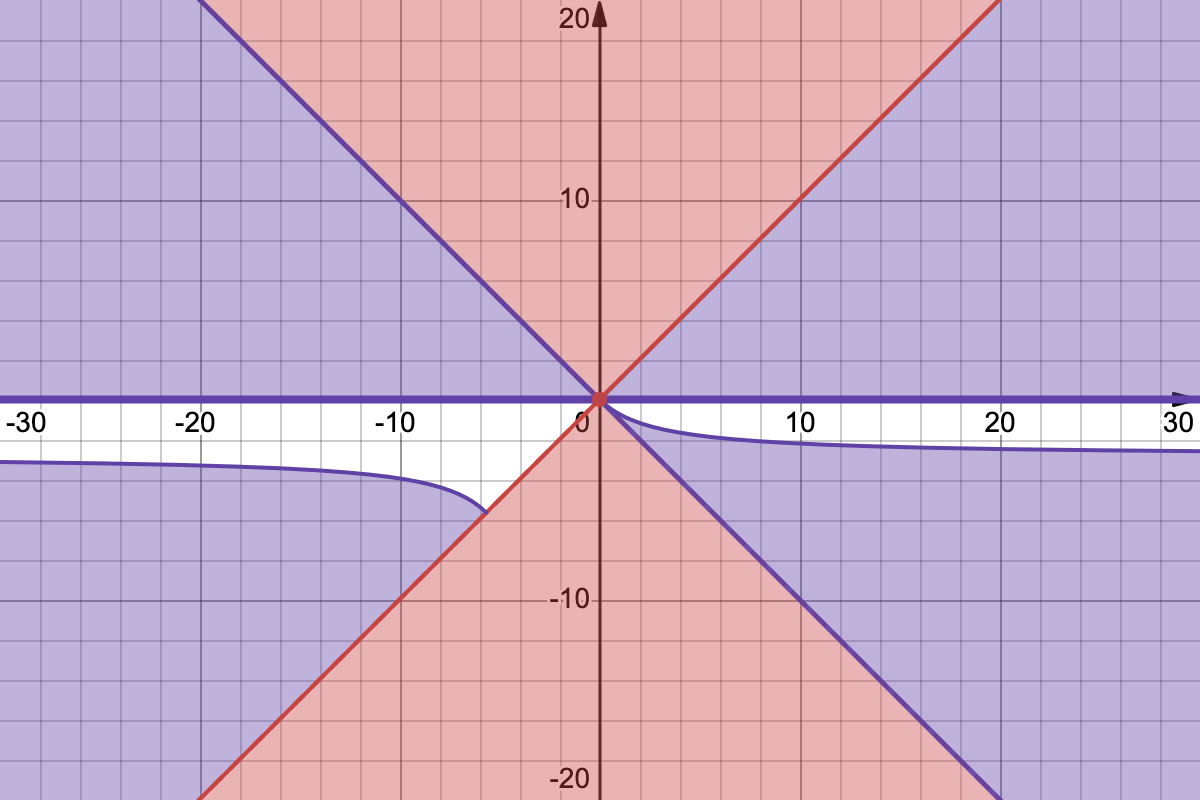}
\caption{The horizontal axis corresponds to $v(1)$ while the vertical axis corresponds to $v(2)$.  
The two white regions correspond to points $\big(v(1),v(2)\big)$ for which there exists a unique  $\theta_0\in (0,\pi)$
leading to a resonant case.}
\label{fig:SN2}
\end{figure}

Before giving the proof, let us briefly comment on this statement with the help of Figures \ref{fig:SN2} and \ref{fig:3D}. 
In Figure \ref{fig:SN2}, the $x$-axis corresponds to the values of $v(1)$ while the $y$-axis
corresponds to the values of $v(2)$. 
The red zones (2 open cones + the diagonal line in quadrants I and III) are the ones we can disregard since either $a< b$ or the system is $1$--periodic. The magenta zones (including the diagonal line in quadrants II and IV, the $x$-axis and the curved boundaries) correspond to combinations of $v(1)$ and $v(2)$ which lead to the generic case, namely to \eqref{eq:k3}. 
The open white region has two connected components and is given by points $\big(v(1),v(2)\big)$ such that there exists a unique $\theta_0\in (0,\pi)$ which verifies $v(2)=-\frac{2\Xi_+(\theta_0)v(1)}{v(1)+2\Xi_+(\theta_0)}$. This condition is equivalent to the condition leading to the resonant case given by \eqref{eq:k1}. In particular, it follows from the equality $2\varrho\Xi_+(\theta_0)=-a^2b^2$ that $\varrho<0$, which means
that either $u_1=u_2=-1$, corresponding to the left white region, or to $u_2=-1$ and $u_1=1$, corresponding to the right white region.

In order to understand the dependence of  $\theta_0$ on $v(1)$ and $v(2)$ 
for the resonant case, we have represented in Figure \ref{fig:3D} the value $\theta_0$ 
on the $z$--axis as function of $v(1)$ on the $x$--axis and of $v(2)$ 
on the $y$--axis. 
Note that $\theta_0$ is only shown on the left white region of Figure \ref{fig:SN2}
which is bounded on the $y$--axis and unbounded on the $x$--axis. 

The following proof is based on the statement \cite[Thm.~3.9]{NRT} which provides
the general formulas. Here, we compute much more explicit results in the special case $N=2$.
Let us stress that most of the notations are directly borrowed from \cite{NRT}. 
In addition, we shall use the notation $\dagger$ for the inverse of an operator defined on its range.

\begin{figure}
\centering
  \includegraphics[width=0.6\textwidth]{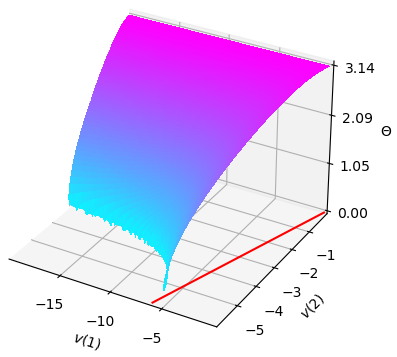}
\caption{The values $\theta_0$ that satisfy the resonant case (in the left white region of Figure \ref{fig:SN2})
is represented, with the value $v(1)$ and $v(2)$ on the $x$ and $y$ axes. The 2 scales are different, and the red line represent the diagonal on this quadrant. The $z$--axis corresponds to $\theta$ in $(0,\pi)$, and the surface represents the values of $\theta_0$, as a function of $v(1)$ and $v(2)$.}
\label{fig:3D}
\end{figure}

\begin{proof}[Proof of Proposition \ref{S(lambda)N=2}]
We start by introducing the new quantities:
\begin{equation*}
\zeta_1^\theta:=\frac{1}{\sqrt{a^2+b^2}}
\left(\begin{smallmatrix}b \\ a\e^{i\theta/2}\end{smallmatrix}\right)
\quad  \zeta_2^\theta:=\frac{1}{\sqrt{a^2+b^2}}
\left(\begin{smallmatrix}b \\ -a\e^{i\theta/2}\end{smallmatrix}\right).
\end{equation*}
It is then easily checked that the following relations hold:
\begin{align*}
& \langle\v \xi_j^\theta,\zeta_{j}^\theta\rangle=0\quad \hbox{and} \quad \langle \v\xi_1^\theta,\v\xi_2^\theta\rangle=b^2-a^2,\\
&\langle \zeta_1^\theta,\u \zeta_1^\theta \rangle=\langle \zeta_2^\theta,\u \zeta_2^\theta\rangle=\frac{\varrho}{a^2+b^2}, \\
& \langle \zeta_1^\theta,\v \xi_2^\theta\rangle=\frac{2ab \e^{i\theta/2}}{\sqrt{a^2+b^2}}=-\langle \zeta_2^\theta,\v \xi_1^\theta\rangle\quad \hbox{and} \quad \langle \u\zeta_1^\theta,\v \xi_1^\theta\rangle=\frac{ab \e^{i\theta/2}}{\sqrt{a^2+b^2}}(u_2-u_1)=-\langle \u \zeta_2^\theta,\v \xi_2^\theta\rangle.
\end{align*}

1) Let us start by considering the threshold $\lambda_1^\theta-2$ for $\theta\in(0,\pi)$. The starting point is the formula provided in \cite[Thm.~3.9.(b)]{NRT}, namely
\begin{equation}\label{eq:S11}
S^\theta(\lambda_1^\theta-2)_{11}
=\P_1^\theta-\P_1^\theta\v\big(I_0(0)+S_0\big)^{-1}\v\;\!\P_{1}^\theta
+\P_1^\theta\v\;\!C_{10}'(0)S_1\big(I_2(0)+S_2\big)^{-1}S_1C_{10}'(0)\v\;\!
\P_{1}^\theta.
\end{equation}
We firstly observe that
\begin{equation*}
I_0(0)=\frac{1}{4}|\v\xi_1^\theta\rangle \langle \v\xi_1^\theta|, \quad  S_0=|\zeta_1^\theta\rangle \langle \zeta_1^\theta|, \quad \beta_2^\theta(\lambda_1^\theta-2)^2=4\Xi_+(\theta), \quad  M_1(0)=\u + \frac{|\v \xi_2^\theta\rangle \langle \v \xi_2^\theta|}{8\Xi_+(\theta)}.
\end{equation*}
and hence
\begin{equation*}
(I_0(0)+S_0)^{-1}=I_0(0)^\dagger+S_0=\frac{4}{(a^2+b^2)^2}|\v\xi_1^\theta\rangle \langle \v\xi_1^\theta|+|\zeta_1^\theta\rangle \langle \zeta_1^\theta|.
\end{equation*}
It then follows that
\begin{equation*}
\P_1^\theta\v(I_0(0)+S_0)^{-1}\v\P_1^\theta=2\P_1^\theta.
\end{equation*}
Thus, if the  third term in the r.h.s.~of \eqref{eq:S11} vanishes, then 
$S^\theta(\lambda_1^\theta-2)_{11}= -\P_1^\theta$. 
This situation corresponds to the generic case.
Our next aim is thus to exhibit situations when the third
term in \eqref{eq:S11} is not trivial.

Let us recall that $S_1$ is defined as the projection on the kernel of $I_1(0)$ (inside the subspace $S_0\C^N$). Thus, in order to deal with the third term, we have to study $I_1(0)$:
\begin{align}\label{firstfactor}
\nonumber & I_1(0) = S_0 M_1(0) S_0 = S_0\u S_0+S_0\frac{|\v \xi_2^\theta\rangle \langle \v \xi_2^\theta|}{8\Xi_+(\theta)}S_0
=\bigg(\langle \zeta_1^\theta ,\u \zeta_1^\theta \rangle+\frac{|\langle\v \xi_2^\theta,\zeta_1^\theta \rangle|^2}{8\Xi_+(\theta)}  \bigg)S_0\\
& =\bigg(\frac{\varrho}{a^2+b^2}+\frac{\frac{4a^2b^2}{a^2+b^2}}{8\Xi_+(\theta)}\bigg)S_0
=\bigg(\frac{2\varrho\Xi_+(\theta)+a^2b^2}{2\Xi_+(\theta)(a^2+b^2)}  \bigg)S_0.
\end{align}
If the first factor in \eqref{firstfactor} does not vanish, then $S_1=0$. 
Clearly, if $S_1=0$, then the generic case takes place. Thus, in order not to be in the generic case the following necessary condition must hold:
\begin{equation}\label{HypA}
2\varrho\Xi_+(\theta)+a^2b^2=0.
\end{equation} 

Let us now study $C_{10}'$ under the condition \eqref{HypA}. Recalling that $S_1=S_0$ we get
\begin{align*}
C'_{10}=&\big(I_0(0)+S_0\big)^{-1}\big(M_1(0)S_1-S_1M_1(0)\big)\big(I_0(0)+S_0\big)^{-1}\\
=&\big(I_0(0)^\dagger+S_0\big)\big(M_1(0)S_0-S_0M_1(0)\big)\big(I_0(0)^\dagger+S_0\big)\\
=&I_0(0)^\dagger M_1(0)S_0-S_0M_1(0)I_0(0)^\dagger.
\end{align*}
Hence, from the explicit expressions computed above one has
\begin{align*}
S_0M_1(0)I_0(0)^\dagger=&|\zeta_1^\theta\rangle\langle\zeta_1^\theta|\Big( \u+\frac{|\v\xi_2^\theta\rangle\langle \v\xi_2^\theta\rangle|}{8\Xi_+(\theta)}  \Big)
 \frac{4}{(a^2+b^2)^2}|\v \xi_1^\theta\rangle\langle \v \xi_1^\theta|\\
=&\frac{1}{(a^2+b^2)^2}\Big(4\langle\u\zeta_1^\theta,\v\xi_1^\theta \rangle+\frac{\langle\zeta_1^\theta , \v \xi_2^\theta \rangle\langle \v\xi_2^\theta,\v\xi_1^\theta\rangle}{2\Xi_+(\theta)}\Big)|\zeta_1^\theta\rangle\langle \v \xi_1^\theta|\\
=&\frac{ab}{(a^2+b^2)^{\frac52}}\Big(4(u_2-u_1)+ \frac{b^2-a^2}{\Xi_+(\theta)}\Big)\e^{i\theta/2} |\zeta_1^\theta\rangle\langle \v \xi_1^\theta|.
\end{align*}
We thus set 
$$
\gamma^{\theta}:=
\frac{ab}{(a^2+b^2)^{\frac52}}\Big(4(u_2-u_1)+ \frac{b^2-a^2}{\Xi_+(\theta)}\Big)
$$
which means that 
$$
S_0M_1(0)I_0(0)^\dagger 
= \gamma^{\theta} \e^{i\theta/2} |\zeta_1^\theta\rangle\langle \v \xi_1^\theta|.
$$
Analogously we also obtain
\begin{equation*}
I_0(0)^\dagger M_1(0)S_0=\gamma^\theta \e^{-i\theta/2} |\v \xi_1^\theta\rangle\langle\zeta_1^\theta |, 
\end{equation*}
and consequently
$$
C'_{10} = \gamma^\theta\Big(
\e^{-i\theta/2} |\v \xi_1^\theta\rangle\langle\zeta_1^\theta |
-  \e^{i\theta/2} |\zeta_1^\theta\rangle\langle \v \xi_1^\theta|\Big).
$$
Thus, $C_{10}'$ vanishes if and only if $\gamma^\theta=0$, and if so, the generic case  takes place again. However, observe that  \eqref{HypA} and the definition of $\varrho$ imply that $b^2>0$, that $\varrho\neq0$, and then that
\begin{equation*}
\Xi_+(\theta)=-\frac{a^2b^2}{2\varrho}.
\end{equation*} 
Observe then that the condition $\gamma^\theta=0$ is equivalent to
\begin{align*}
 	& 4(u_2-u_1)+\frac{b^2-a^2}{\Xi_+(\theta)}=0\\
\Leftrightarrow\  & 2(u_2-u_1)a^2b^2 +(a^2-b^2)(a^2u_2+b^2u_1)=0\\
\Leftrightarrow\  & u_2a^2b^2-u_1a^2b^2+a^4u_2-b^4u_1=0\\
\Leftrightarrow\  & (b^2+a^2)(u_2a^2-u_1b^2)=0.
\end{align*}
However, from the fact that $v$ is assumed to be
$2$-periodic and not $1$-periodic, one infers that $\gamma^\theta=0$
does not take place under assumption \eqref{HypA}.

Now, it follows that under \eqref{HypA} one has
$S_1=S_0$, $I_1(0)= S_0M_1(0)S_0=0$, and then
\begin{align*}
I_2(0)=&-S_1 M_1(0)(I_0(0)+S_0)^{-1}M_1(0)S_1 \\
=&-S_0 M_1(0)I_0(0)^\dagger M_1(0)S_0 \\
=&-\gamma^{\theta}\e^{i\theta/2} |\zeta_1^\theta\rangle\langle \v \xi_1^\theta| \ I_0(0)\  \gamma^\theta \e^{-i\theta/2} |\v \xi_1^\theta\rangle\langle\zeta_1^\theta |\\
=&\frac{-\big(\gamma^\theta(a^2+b^2)\big)^2}4 S_0.
\end{align*}
Since $S_2$ corresponds to the projection on the kernel of $I_2(0)$, one infers that
\begin{equation}\label{eq:S_2=0}
S_2=0\quad \text{ and }\quad (I_2(0)+S_2)^{-1}=\frac{-4}{\big(\gamma^\theta(a^2+b^2)\big)^2} S_0.
\end{equation}

We have now all the ingredients for computing the third term in \eqref{eq:S11},
still under the conditions provided by \eqref{HypA}\;\!:
\begin{align*}
&\P_1^\theta\v\;\!C_{10}'(0)S_1\big(I_2(0)+S_2\big)^{-1}S_1C_{10}'(0)\v\;\!
\P_{1}^\theta \\
& = \frac{| \xi_1^\theta\rangle\langle\v\xi_1^\theta|}{2} 
 \gamma^\theta\Big( \e^{-i\theta/2} |\v \xi_1^\theta\rangle\langle\zeta_1^\theta |
-  \e^{i\theta/2} |\zeta_1^\theta\rangle\langle \v \xi_1^\theta|\Big)
\frac{-4S_0 }{\big(\gamma^\theta(a^2+b^2)\big)^2}   \\
& \qquad \times 
 \gamma^\theta\Big( \e^{-i\theta/2} |\v \xi_1^\theta\rangle\langle\zeta_1^\theta |
-  \e^{i\theta/2} |\zeta_1^\theta\rangle\langle \v \xi_1^\theta|\Big)
\frac{|\v\xi_1^\theta\rangle\langle\xi_1^\theta|}{2}\\
&=\frac{| \xi_1^\theta\rangle\langle\v\xi_1^\theta|}{2} \e^{-i\theta/2} |\v \xi_1^\theta\rangle\langle\zeta_1^\theta |\frac{4S_0 }{(a^2+b^2)^2}  \e^{i\theta/2} |\zeta_1^\theta\rangle\langle \v \xi_1^\theta|\frac{|\v\xi_1^\theta\rangle\langle\xi_1^\theta|}{2}\\
&=\frac{1}{(a^2+b^2)^2}| \xi_1^\theta\rangle \|\v\xi_1^\theta \|^2 
\langle\zeta_1^\theta,S_0\zeta_1^\theta\rangle \|\v \xi_1^\theta\|^2 
\langle\xi_1^\theta| \\
&=2\P_1^\theta.
\end{align*}
Thus, whenever  \eqref{HypA} hold, one has
$S^\theta(\lambda_1^\theta-2)=\P_1^\theta$.

2) We now consider the threshold $\lambda_2^\theta-2$ for $\theta\in (0,\pi)$. 
For this value one has
\begin{equation*}
I_0(0)=\frac{1}{4}|\v\xi_2^\theta\rangle\langle\v\xi_2^\theta|, \quad  S_0=|\zeta_2^\theta\rangle\langle\zeta_2^\theta|,  \quad 
\beta_1^\theta(\lambda_2^\theta-2)^2=4\Xi_{-}(\theta), \quad M_1(0)=u+\frac{i|\v\xi_1^\theta\rangle\langle\v\xi_1^\theta|}{8\Xi_-(\theta)}.
\end{equation*}
We can then compute
\begin{equation*}
I_1(0)=\Big(\langle\zeta_2^\theta, \u \zeta_2^\theta\rangle+i\frac{|\langle\zeta_2^\theta,\v\xi_1^\theta\rangle|^2}{8\Xi_-(\theta)}\Big)S_0
=\Big(\frac{\varrho}{a^2+b^2}+i\frac{a^2b^2}{2\Xi_-(\theta)(a^2+b^2)} \Big)S_0.
\end{equation*}
A quick inspection at the first factor shows that it never vanishes, which means that $S_1=0$. Then, for this threshold, the expression for $S^\theta(\lambda_2^\theta-2)_{22}$ can be computed by using \cite[Thm.~3.9.(b)]{NRT}.
Note that the off diagonal terms vanish, as shown in \cite[Thm.~3.9.(b)]{NRT}. For the expression $S^\theta(\lambda_2^\theta-2)_{22}$ we obtain
\begin{align*}
S^\theta(\lambda_2^\theta-2)_{22}&=\P_2^\theta-\P_2^\theta\v(I_0(0)+S_0)^{-1} \v \P_2^\theta \\
&=\P_2^\theta-\frac14 \;\!\frac{4}{(a^2+b^2)^2}|\langle \v\xi_2^\theta, \v\xi_2^\theta\rangle|^2 |\xi_2^\theta\rangle\langle \xi_2^\theta|\\
&=-\P_2^\theta,
\end{align*}
which concludes the proof.
\end{proof}

Up to now, only the values of the scattering matrix at thresholds have been provided.
Additional information can be deduced from this computation, but some notations have
to be introduced.
Firstly, let $G:\h\to\C^2$ be the bounded operator defined on $\f\in\h$ by 
\begin{equation*}
(G\f)_j:=\v_{jj}\int_0^\pi \f_j(\omega)\;\!\tfrac{\d\omega}\pi
=|v(j)|^{1/2}\int_0^\pi \f_j(\omega)\;\!\tfrac{\d\omega}\pi,
\end{equation*}
and set for $\lambda, \varepsilon\in \R$ and $\varepsilon \neq 0$
\begin{equation*}
M^\theta(\lambda+i\varepsilon):=\big(\u+G(H_0^\theta-\lambda-i \varepsilon)^{-1}G^*\big)^{-1}.
\end{equation*}
A precise (and rather long) asymptotic expansion of this has been provided in \cite[Prop.~3.4]{NRT}
as $\varepsilon$ approaches $0$ for each fixed $\lambda$ in the spectrum of $H^\theta$.
Based on the previous proof, more precise information can be obtained in our current
framework, namely for $N=2$.

\begin{Corollary}\label{cor:resonantcase}
Let $\theta \in(0,\pi)$ and consider $a\geq b \geq 0$ with $a>0$. Let us assume that
\begin{equation}\label{eq:genericassumption}
\Xi_+(\theta)\neq\frac{a^2b^2}{2\varrho}.
\end{equation}
Then the map $\kappa \mapsto M^\theta(\lambda^\theta_1-2-\kappa^2)$ admits an holomorphic extension in some sufficiently small neighbourhood of $0$ for 
$-\kappa^2\in \overline{\C_+}$ (the closure of $\C_+$).
If, in contrast
\begin{equation}\label{eq:resonanceassumption}
\Xi_+(\theta)=\frac{a^2b^2}{2\varrho},
\end{equation}
the map $\kappa \mapsto  M^\theta(\lambda^\theta_1-2-\kappa^2)$ admits an meromorphic extension in some sufficiently small neighbourhood of $0$ in 
$ \overline{\C_+}$, with its unique pole located at $0$ and of order $1$ (in the variable $\kappa$).
\end{Corollary}

Before giving the proof, let us stress that the pole of order $1$ (in the variable $\kappa$) does not correspond to an eigenvalue of $H^\theta$ at $\lambda^\theta_1-2$, but corresponds to a resonance. Indeed, an eigenvalue would necessitate a pole of order $2$ in $\kappa$, 
as for the resolvent $\big(H^\theta-(\lambda^\theta_1-2-\kappa^2)\big)^{-1}$.  
In particular, this statement rules out the existence of an eigenvalue at the bottom
of the continuous spectrum of $H^\theta$.

\begin{proof}
We first observe that \eqref{eq:genericassumption} corresponds to the generic case,  and hence to $S_1=0$. Then, by examining the expansion of $M^\theta(\lambda^\theta_1-2-\kappa^2)$ provided in \cite[(3.17)]{NRT}, we infers that $S_1=0$ leads to a regular expression, without the singularities in $\tfrac1k$ and in $\tfrac1{k^2}$. 
In contrast, if we assume \eqref{eq:resonanceassumption} we have shown in the previous proof that $S_1\neq 0$ but $S_2=0$,  see \eqref{eq:S_2=0}. Hence, by looking again at  \cite[(3.17)]{NRT}, we deduce that the singularity in $\tfrac1{k^2}$ vanishes while the one in $\tfrac1k$ does not. 
\end{proof}

\subsection{The scattering matrix \texorpdfstring{$S^\theta(\lambda)$}{S(lambda)}} \label{subsec:inside}

In this section we provide explicit formulas for $S^\theta(\lambda)$. 
The starting point is a stationary representation of the scattering matrix.
Namely, let us also define for $j,j'\in\{1,\dots,N\}$ the operator
$\delta_{jj'}\in\B\big(\P_{j'}^\theta\;\!\C^N;\P_j^\theta\;\!\C^N\big)$ by
$\delta_{jj'}:=1$ if $j=j'$ and $\delta_{jj'}:=0$ otherwise. Then,  for
$
\lambda\in\big(I^\theta_j\cap I^\theta_{j'}\big)
\setminus\big(\T^\theta\cup\sigma_{\rm p}(H^\theta)\big)
$
the channel scattering matrix
$S^\theta(\lambda)_{jj'}:=\P^\theta_jS^\theta(\lambda)\P^\theta_{j'}$ satisfies the
formula
\begin{equation}\label{eq_S_lambda}
S^\theta(\lambda)_{jj'}
=\delta_{jj'}-2i\;\!\beta_j^\theta(\lambda)^{-1}\P_j^\theta\v\;\!
M^\theta(\lambda+i0)\v\;\!\P_{j'}^\theta\;\!\beta_{j'}^\theta(\lambda)^{-1},
\end{equation}
with 
$$
M^\theta(\lambda+i0) =\big(\u+G(H_0^\theta-\lambda-i0)^{-1}G^*\big)^{-1}.
$$
The operator $M^\theta(\lambda+i0)$ belongs to $\B(\C^2)$ for
$\lambda \notin \big(\tau^\theta \cup\sigma_{\rm p}(H^\theta) \big)$
and by \cite[Lem.~3.1]{NRT} one has
\begin{equation}
\u + G(H_0^\theta-\lambda-i0)^{-1}G^*=
\begin{cases}
\u+\frac{\v\P_1^\theta\v}{\beta_1^\theta(\lambda)^2}+\frac{\v\P_2^\theta\v}{\beta_2^\theta(\lambda)^2}& \hbox{ for }\ \lambda<\lambda_1^\theta-2, \\
\u+i\frac{\v\P_1^\theta\v}{\beta_1^\theta(\lambda)^2}+\frac{\v\P_2^\theta\v}{\beta_2^\theta(\lambda)^2}& \hbox{ for }\ \lambda \in (\lambda_1^\theta-2,\lambda_2^\theta-2),
\\
\u+i\frac{\v\P_1^\theta\v}{\beta_1^\theta(\lambda)^2}+i\frac{\v\P_2^\theta\v}{\beta_2^\theta(\lambda)^2}&\hbox{ for }\ \lambda \in (\lambda_2^\theta-2, \lambda_1^\theta+2),
\\
\u-\frac{\v\P_1^\theta\v}{\beta_1^\theta(\lambda)^2}+i\frac{\v\P_2^\theta\v}{\beta_2^\theta(\lambda)^2}&\hbox{ for }\ \lambda \in (\lambda_1^\theta+2,\lambda_2^\theta+2),
\\
\u-\frac{\v\P_1^\theta\v}{\beta_1^\theta(\lambda)^2}-\frac{\v\P_2^\theta\v}{\beta_2^\theta(\lambda)^2}&\hbox{ for }\ \lambda >\lambda_2^\theta+2.
\end{cases}
\end{equation}
Hence we shall look for an explicit expression for 
\begin{equation*}
M^\theta(\lambda+i0)=\Big(\u+c_1\frac{\v\P_1^\theta\v}{\beta_1^\theta(\lambda)^2}+c_2\frac{\v\P_2^\theta\v}{\beta_2^\theta(\lambda)^2}\Big)^{-1}
\end{equation*}
for $c_j\in\{1,-1,i\}$.
For convenience we set
\begin{equation*}
T_\pm^\theta(\lambda):=\frac{c_1}{\beta_1^\theta(\lambda)^2}\pm\frac{c_2}{\beta_2^\theta(\lambda)^2}.
\end{equation*}

\begin{Lemma}\label{lemmaM1invertible}
For all $\lambda\in (\lambda_1^\theta-2,\lambda_2^\theta+2)\backslash \{\lambda_2^\theta-2,\lambda_1^\theta+2\}$ the operator $M^\theta(\lambda+i0)$ belongs to $\B(\C^2)$ and one has
\begin{equation}\label{eq:inv}
\v M^\theta(\lambda+i0)\v= \frac{1}{Q^\theta(\lambda)} \begin{pmatrix}
u_2a^2+\frac{a^2b^2}2 T_+^\theta(\lambda) & \frac{a^2b^2 \e^{-i\frac{\theta}2}}2T_-^\theta(\lambda)\\
\frac{a^2b^2 \e^{i\frac{\theta}2}}2T_-^\theta(\lambda) & u_1b^2+\frac{a^2b^2}2 T_+^\theta(\lambda)
\end{pmatrix}
\end{equation}
with
\begin{equation}\label{eq:det}
Q^\theta(\lambda):=\frac{2u_1u_2\beta_1^\theta(\lambda)^2\beta_2^\theta(\lambda)^2+\big(c_1\beta_2^\theta(\lambda)^2+c_2\beta_1^\theta(\lambda)^2\big)(u_2a^2+u_1b^2)+2a^2b^2c_1c_2}{2\beta_1^\theta(\lambda)^2\beta_2^\theta(\lambda)^2} .
\end{equation}
\end{Lemma}

\begin{proof}
We first observe that
\begin{equation*}
\u + G(H_0^\theta-\lambda-i0)^{-1}G^*=\begin{pmatrix}
u_1+\frac{a^2}2T_+^\theta(\lambda) & -\frac{ab \e^{-i\frac{\theta}{2}}}2T_-^\theta(\lambda)\\
-\frac{ab \e^{i\frac{\theta}{2}}}2T_-^\theta(\lambda)& u_2+\frac{b^2}2T_+^\theta(\lambda) 
\end{pmatrix}.
\end{equation*}
Expressions \eqref{eq:inv} and \eqref{eq:det} are then readily obtained. 

In only remains to check that $Q^\theta(\lambda)\neq 0$, which corresponds to the absence of embedded eigenvalues. 
Since $\lambda \not \in \T^\theta$, the condition $Q^\theta(\lambda)= 0$ is equivalent to
\begin{equation}\label{eq:testM}
2u_1u_2\beta_1^\theta(\lambda)^2\beta_2^\theta(\lambda)^2+\big(c_1\beta_2^\theta(\lambda)^2+c_2\beta_1^\theta(\lambda)^2\big)(u_2a^2+u_1b^2)+2a^2b^2c_1c_2=0.
\end{equation}

\paragraph{Case: $\lambda \in (\lambda_1^\theta-2, \lambda_2^\theta-2)$.}
One has $c_1=i$ and $c_2=1$, and then \eqref{eq:testM} is equivalent to the system
\begin{equation*}
\begin{cases}
u_1u_2\beta_2^\theta(\lambda)^2+\frac12(u_2a^2+u_1b^2)=0\\
\frac12(u_2a^2+u_1b^2)\beta_2^\theta(\lambda)^2+a^2b^2=0.
\end{cases}
\end{equation*}
From the second equation we need $(u_2a^2+u_1b^2)<0$ and hence we have either $u_2=u_1=-1$, or $u_2=-1$ and $u_1=1$. From the first equation we need $u_1u_2=1$ and hence we can disregard the latter option. The new system becomes 
\begin{equation*}
\begin{cases}
\beta_2^\theta(\lambda)^2-\frac12(a^2+b^2)=0\\
-\frac12(a^2+b^2)\beta_2^\theta(\lambda)^2+a^2b^2=0.
\end{cases}
\end{equation*}
Isolating $\beta_2^\theta(\lambda)^2$ we obtain 
$\frac12(a^2+b^2)=\frac{2a^2b^2}{a^2+b^2}$
which is equivalent to $a^2=b^2$. As a result, $v(1)=v(2)$, which is excluded.  

\paragraph{Case: $\lambda \in (\lambda_2^\theta-2,\lambda_1^\theta+2)$.}
In this case one has $c_1=c_2=i$ and \eqref{eq:testM} becomes
\begin{equation*}
\begin{cases}
u_1u_2\beta_1^\theta(\lambda)^2\beta_2^\theta(\lambda)^2-a^2b^2=0\\
\big(\beta_2^\theta(\lambda)^2+\beta_1^\theta(\lambda)^2\big)(u_2a^2+u_1b^2)=0.
\end{cases}
\end{equation*}
From the second equation we infer that $a^2=b^2$ and $u_1u_2=-1$, but this immediately contradicts the first equation.

\paragraph{Case: $\lambda \in (\lambda_1^\theta+2,\lambda_2^\theta+2)$.}
In this case we have $c_1=-1$ and $c_2=i$ \eqref{eq:testM} becomes
\begin{equation*}
\begin{cases}
u_1u_2\beta_1^\theta(\lambda)^2-\frac12(u_2a^2+u_1b^2)=0\\
\frac12(u_2a^2+u_1b^2)\beta_1^\theta(\lambda)^2-a^2b^2=0
\end{cases}
\end{equation*}
which can be treated as the first case.
\end{proof}

As shown in the previous proof, let us emphasise the following statement:

\begin{Corollary}
For any $\theta \in (0,\pi)$, the operator $H^\theta$ has no embedded eigenvalue. Furthermore, a necesary and sufficent condition for $\lambda$ to be an eigenvalue is given for $\lambda<\lambda_1^\theta-2$ by
\begin{equation}\label{eq:eigenvaluebelow}
2u_1u_2\beta_1^\theta(\lambda)^2\beta_2^\theta(\lambda)^2+\big(\beta_2^\theta(\lambda)^2+\beta_1^\theta(\lambda)^2\big)(u_2a^2+u_1b^2)+2a^2b^2=0
\end{equation}
and for $\lambda > \lambda^\theta_2+2$ by
\begin{equation}\label{eq:eigenvalueabove}
2u_1u_2\beta_1^\theta(\lambda)^2\beta_2^\theta(\lambda)^2-\big(\beta_2^\theta(\lambda)^2+\beta_1^\theta(\lambda)^2\big)(u_2a^2+u_1b^2)+2a^2b^2=0.
\end{equation}
\end{Corollary}

Based on \Cref{lemmaM1invertible} we can obtain simple expressions for 
$S_{jj'}^\theta(\lambda)$. In the next statement, we concentrate on the diagonal elements
$S_{jj}^\theta(\lambda)$ since the off-diagonal element will not be necessary later on. 
We define a scalar valued function by the relation 
$S^\theta(\lambda)_{jj}=:\s_{jj}^\theta(\lambda)\P_j^\theta$. 

Let us also introduce new expressions which will be used subsequently:
\begin{align}
& A_1^\theta(\lambda):=\beta_1^\theta(\lambda)^2\big(2u_1u_2\beta_2^\theta(\lambda)^2+\varrho\big), &\quad &B_1^\theta(\lambda):=2a^2b^2+\varrho\beta_2^\theta(\lambda)^2, \label{A1B1}\\
&A_2^\theta(\lambda):=2u_1u_2\beta_1^\theta(\lambda)^2\beta_2^\theta(\lambda)^2-2a^2b^2,  & \quad  & B_2^\theta(\lambda):=\varrho\big(\beta_1^\theta(\lambda)^2+\beta_2^\theta(\lambda)^2\big),\label{C2D2}\\
& A_3^\theta(\lambda):=\beta_2^\theta(\lambda)^2\big(2u_1u_2\beta_1^\theta(\lambda)^2-\varrho\big), & \quad &B_3^\theta(\lambda):=-2a^2b^2+\varrho\beta_1^\theta(\lambda)^2 \label{A3B3} .
\end{align}

\begin{Proposition}\label{prop:grosses_expressions}
The following formulae hold:
\begin{equation}\label{eq:s11}
\s_{11}^\theta(\lambda)=\begin{cases}
\frac{A_1^\theta(\lambda)-iB_1^\theta(\lambda)}{A_1^\theta(\lambda)+iB_1^\theta(\lambda)}&\text{ for }\ \lambda \in (\lambda^\theta_1-2,\lambda^\theta_2-2),\\[2mm]
1+\frac{4a^2b^2-i2\varrho\beta_2^\theta(\lambda)^2}{A_2^\theta(\lambda)+iB_2^\theta(\lambda)}&\text{ for }\ \lambda \in (\lambda^\theta_2-2, \lambda^\theta_1+2),
\end{cases}
\end{equation}
and
\begin{equation}\label{eq:s22}
\s_{22}^\theta(\lambda)=\begin{cases}
1+\frac{4a^2b^2-i2\varrho\beta_1^\theta(\lambda)^2}{A_2^\theta(\lambda)+iB_2^\theta(\lambda)}&\text{ for }\ \lambda \in (\lambda^\theta_2-2,\lambda^\theta_1+2),\\[2mm]
\frac{A_3^\theta(\lambda)-iB_3^\theta(\lambda)}{A_3^\theta(\lambda)+iB_3^\theta(\lambda)}&\text{ for }\ \lambda \in (\lambda^\theta_1+2,\lambda^\theta_2+2).
\end{cases}
\end{equation}
\end{Proposition}

\begin{proof}
We start with the computation for $\s^\theta_{11}(\lambda)$.
Let us firstly observe from \eqref{eq_S_lambda} that
\begin{equation*}
S^\theta(\lambda)_{11}=\P_1^\theta -\frac{2i}{\beta_1^\theta(\lambda)^2}\P_1^\theta \v M^\theta(\lambda+i0)\v \P_1^\theta
=\Big(1-\frac{i}{\beta_1^\theta(\lambda)^2}
\big\langle\xi^\theta_1,\v M^\theta(\lambda+i0)\v \xi_1^\theta\big\rangle\Big)\P_1^\theta.
\end{equation*}
By using Lemma \ref{lemmaM1invertible} and the explicit form of $\xi^\theta_1$ one easily infers that 
\begin{align*}
\big\langle \xi^\theta_1 , \v M(\lambda+i0) \v \xi_1^\theta \big\rangle 
=\frac{1}{Q^\theta(\lambda)}\Big(a^2b^2\big(T^\theta_+(\lambda)-T^\theta_-(\lambda)\big)+\varrho\Big),
\end{align*}
and by a few additional computations
\begin{align}
\s_{11}^\theta(\lambda)=&1-i\frac{a^2b^2\big(T^\theta_+(\lambda)-T^\theta_-(\lambda)\big)+\varrho}{\beta_1^\theta(\lambda)^2 Q^\theta(\lambda)}\nonumber \\
=&1-i\frac{4a^2b^2c_2+2\beta_2^\theta(\lambda)^2\varrho}{2u_1u_2\beta_1^\theta(\lambda)^2\beta_2^\theta(\lambda)^2+\big(c_1\beta_2^\theta(\lambda)^2+c_2\beta_1^\theta(\lambda)^2\big)\varrho+2a^2b^2c_1c_2}.\label{eq:s11general}
\end{align}

Let us now recall that for $\lambda \in (\lambda_1^\theta-2,\lambda_2^\theta-2)$ we have $c_1=i$ and $c_2=1$. It follows that for such $\lambda$ we have
\begin{align*}
\s_{11}^\theta(\lambda)=&1-i\frac{4a^2b^2+2\beta_2^\theta(\lambda)^2\varrho}{2u_1u_2\beta_1^\theta(\lambda)^2\beta_2^\theta(\lambda)^2+\big(i\beta_2^\theta(\lambda)^2+\beta_1^\theta(\lambda)^2\big)\varrho+i2a^2b^2}\\
=&1+\frac{i\big(-4a^2b^2-2\beta_2^\theta(\lambda)^2\varrho\big)}{\beta_1^\theta(\lambda)^2\big(2u_1u_2\beta_2^\theta(\lambda)^2+\varrho\big)+i\big(\beta_2^\theta(\lambda)^2\varrho+2a^2b^2\big)},
\end{align*}
from where the first formula of \eqref{eq:s11} follows. For $\lambda \in (\lambda_2^\theta-2,\lambda_1^\theta+2)$ the result can also be directly obtained from \eqref{eq:s11general},  recalling that in this case $c_1=c_2=i$.

We now sketch the computation for $S^\theta(\lambda)_{22}$. From  \eqref{eq_S_lambda}
one has
\begin{equation*}
S^\theta(\lambda)_{22}=\P_2^\theta -\frac{2i}{\beta_2^\theta(\lambda)^2}\P_2^\theta \v M^\theta(\lambda+i0)\v \P_2^\theta
=\Big(1-\frac{i}{\beta_2^\theta(\lambda)^2}
\big\langle\xi^\theta_2,\v M^\theta(\lambda+i0)\v \xi_2^\theta\big\rangle\Big)\P_2^\theta.
\end{equation*}
Recalling that $ \xi_2^\theta=\left(\begin{smallmatrix}\e^{i\theta/2} \\ \e^{i\theta}\end{smallmatrix}\right)$ and using \eqref{eq:inv} we easily obtain 
\begin{equation*}
\s_{22}^\theta(\lambda)
=1-i\frac{4a^2b^2c_1+2\beta_1^\theta(\lambda)^2\varrho}{2u_1u_2\beta_1^\theta(\lambda)^2\beta_2^\theta(\lambda)^2+\big(c_1\beta_2^\theta(\lambda)^2+c_2\beta_1^\theta(\lambda)^2\big)\varrho+2a^2b^2c_1c_2}.
\end{equation*}
As before, the first formula of \eqref{eq:s22} can be obtained directly by replacing $c_1=c_2=i$. For the second one, we replace $c_1=-1$ and $c_2=i$ and compute 
\begin{align}
\s_{22}^\theta(\lambda)=&1-i\frac{-4a^2b^2+2\beta_1^\theta(\lambda)^2\varrho}{2u_1u_2\beta_1^\theta(\lambda)^2\beta_2^\theta(\lambda)^2-\beta_2^\theta(\lambda)^2\varrho+i(\beta_1^\theta(\lambda)^2\varrho-2a^2b^2)}\nonumber \\
=&\frac{2u_1u_2\beta_1^\theta(\lambda)^2\beta_2^\theta(\lambda)^2-\beta_2^\theta(\lambda)^2\varrho-i\big(\beta_1^\theta(\lambda)^2\varrho-2a^2b^2\big)}{2u_1u_2\beta_1^\theta(\lambda)^2\beta_2^\theta(\lambda)^2-\beta_2^\theta(\lambda)^2\varrho+i\big(\beta_1^\theta(\lambda)^2\varrho-2a^2b^2\big)},
\end{align}
which leads to the expected result.
\end{proof}

For completeness, let us show that the previous expressions also allow us to compute the values
of the scattering matrix at thresholds. For the thresholds $\lambda^\theta_1-2$ and $\lambda^\theta_2-2$, these values were already presented in Proposition \ref{S(lambda)N=2}.
We therefore only provide the missing information, namely the values at $\lambda_1^\theta+2$ and at $\lambda^\theta_2+2$. Observe however that the treatement presented in Section \ref{subsec:thresholds} contains extra information: it shows for example that eigenvalues
located at thresholds do not exist, while the argument presented below does not 
rule out this possibility.

\begin{Proposition}\label{S(lambda)N=2bis}
Let $\theta \in(0,\pi)$ and consider $a\geq b \geq 0$ with $a>0$.
Then the following equalities hold:
\begin{numcases}{S^\theta(\lambda_2^\theta+2)_{22}=}
 \P_2^{\theta} & if\  $2\varrho\Xi_+(\theta)-a^2b^2=0$, \label{eq:k1bis} \\
 -\P_2^{\theta} & otherwise, \label{eq:k3bis}
\end{numcases}
and
\begin{equation*}   
S^\theta(\lambda_1^\theta+2)_{11}=-\P_1^\theta. 
\end{equation*}
\end{Proposition}

\begin{proof}
Starting from \eqref{eq:s11} for $\lambda \in (\lambda_2^\theta-2, \lambda^\theta_1+2)$
and by observing that $\lim_{\lambda \nearrow \lambda^\theta_1+2}\beta^\theta_1(\lambda)^2=0$, 
one easily infers that $\lim_{\lambda \nearrow \lambda^\theta_1+2}\s^\theta_{11}(\lambda)=-1$.
This result leads direly to $S^\theta(\lambda_1^\theta+2)_{11}=-\P_1^\theta$.

On the other hand, by starting from  \eqref{eq:s22} for $\lambda \in (\lambda_1^\theta+2, \lambda^\theta_1+2)$ and by observing that $\lim_{\lambda \nearrow \lambda^\theta_2+2}
\beta^\theta_2(\lambda)^2=0$, one has to be more careful for the limit
$\lim_{\lambda \nearrow \lambda^\theta_2+2}\s^\theta_{22}(\lambda)$.
If $B^\theta_3(\lambda^\theta_2+2)\neq 0$, then
$\lim_{\lambda \nearrow \lambda^\theta_2+2}\s^\theta_{22}(\lambda)=-1$.
On the other hand, if  $B^\theta_3(\lambda^\theta_2+2)= 0$, observe that
for $\varepsilon>0$ one has
$A^\theta_3(\lambda^\theta_2+2-\varepsilon)=O(\varepsilon^{1/2})$ as $\varepsilon \searrow 0$, 
while $B_3(\lambda^\theta_2+2-\varepsilon)=O(\varepsilon)$ as $\varepsilon \searrow 0$. 
In this case, we then 
infer that $\lim_{\lambda \nearrow \lambda^\theta_2+2}\s^\theta_{22}(\lambda)=1$.
It only remains to observe that 
$$
B^\theta_3(\lambda^\theta_2+2)= 0 \Longleftrightarrow
2\varrho\Xi_+(\theta)-a^2b^2=0,
$$
which leads to \eqref{eq:k1bis}.
\end{proof}

Let us briefly compare Proposition \ref{S(lambda)N=2} with Proposition \ref{S(lambda)N=2bis}. 
By comparing \eqref{eq:k1} and  \eqref{eq:k1bis} one deduces that it is impossible to have simultaneously a resonance at  $\lambda_1^\theta-2$ and at $\lambda_2^\theta+2$. 
Here, simultaneously means for the same potential $V$.
In fact, looking at Figure \ref{fig:SN2} one infers that the regions compatible with a resonance at $\lambda_2^\theta+2$ are the image of the white regions obtained by a symmetry with respect to the origin.
 
\subsection{Computations of some total variations}\label{subsec:windinghorizontal}

In this section, we compute some winding numbers for some values of $u_1, u_2, a^2, b^2$, and
$\theta$.
We know from Theorem \ref{thm:Adam} and Proposition \ref{S(lambda)N=2} that 
\begin{equation}\label{levinsonformularesonances}
\# \sigma_{\rm p}(H^\theta)=2-\frac{\#\{ \text{resonances of }H^\theta \}}2+ \Var \big(\lambda \to \det S^\theta(\lambda)\big).
\end{equation}
The total variation is computed as the sum of the three partial contributions on 
$(\lambda_1^\theta-2,\lambda_2^\theta-2)$, $(\lambda_2^\theta-2,\lambda_1^\theta+2)$, 
and on $(\lambda_1^\theta+2,\lambda_2^\theta+2)$. 
We stress that the increase of the variation is counted clockwise.

We start by a simple statement providing the expression for $\det S^\theta(\lambda)$.

\begin{Lemma}\label{lem:calculodetslambda}
For any $\theta \in (0,\pi)$ the following equalities hold:
\begin{equation*}
\det S^\theta(\lambda)=\begin{cases}
\s^\theta_{11}(\lambda)& \hbox{ for }\  \lambda \in (\lambda_1^\theta-2,\lambda_2^\theta-2), \\[2mm]
\frac{\s^\theta_{22}(\lambda)}{\overline{\s^\theta_{11}(\lambda)}}&\hbox{ for }\ \lambda \in (\lambda_2^\theta-2,\lambda_1^\theta+2), \\[2mm]
\s^\theta_{22}(\lambda)&\hbox{ for }\ \lambda \in (\lambda_1^\theta+2,\lambda_2^\theta+2),
\end{cases}
\end{equation*}
where the r.h.s.~should be understood as the limit
$
\lim_{\tilde{\lambda}\to\lambda}\frac{\s^\theta_{22}(\tilde{\lambda})}{\overline{\s^\theta_{11}(\tilde{\lambda})}}$
if $\s^\theta_{11}(\lambda)=0$.
\end{Lemma}

\begin{proof}
We only need to check the statement for $\lambda \in (\lambda_2^\theta-2,\lambda_1^\theta+2)$, 
and this statement can be deduced easily from the general form of a unitary $2\times 2$ matrix,  as long as $\s_{11}^\theta(\lambda )\neq 0$.
Then, we only need to check when $\s^\theta_{11}(\lambda)$ vanishes. 
From \eqref{eq:s11} we see that this only occurs when
\begin{equation*}
\begin{cases}
4a^2b^2=-A_2^\theta(\lambda)\\
2\varrho\beta_2^\theta(\lambda)^2=B_2^\theta(\lambda)
\end{cases}
\Longleftrightarrow
\begin{cases}
a^2b^2=-u_1u_2\beta_1^\theta(\lambda)^2\beta_2^\theta(\lambda)^2\\
\beta_2^\theta(\lambda)^2=\beta_1^\theta(\lambda)^2 \quad \hbox{ or }\quad \varrho=0.
\end{cases}.
\end{equation*}
Let us first suppose that $\varrho=0$ and notice that the first equation becomes
\begin{equation*}
a^4=\beta_1^\theta(\lambda)^2\beta_2^\theta(\lambda)^2
\end{equation*}
which can be satisfied for at most two values of $\lambda$. If $\varrho\neq0$
one observes that $\beta_2^\theta(\lambda)^2=\beta_1^\theta(\lambda)^2$ only for $\lambda=\frac{1}{2}(\lambda_2^\theta+\lambda_1^\theta)=0$.
The result follows then from the continuity of the map $\lambda \mapsto S^\theta(\lambda)$ outside of $\T^\theta$.
\end{proof}

The computation of $\Var \big(\lambda \to \det S^\theta(\lambda)\big)$ for 
$\lambda \in (\lambda_1^\theta-2,\lambda_2^\theta-2)$ and for $\lambda \in (\lambda_1^\theta+2,\lambda_2^\theta+2)$ can be directly inferred from Lemma \ref{lem:calculodetslambda} and from 
\eqref{eq:s11} and \eqref{eq:s22}.
Let us now deduce a convenient expression for $\lambda \in (\lambda_2^\theta-2,\lambda_1^\theta+2)$.

\begin{Lemma}
For $\lambda \in (\lambda_2^\theta-2,\lambda_1^\theta+2)$ one has
\begin{equation}\label{eq:Sint}
\frac{\s^\theta_{22}(\lambda)}{\overline{\s^\theta_{11}(\lambda)}}
=\frac{A^\theta_2(\lambda)-iB^\theta_2(\lambda)}{A^\theta_2(\lambda)+iB^\theta_2(\lambda)}.
\end{equation}
\end{Lemma}

\begin{proof}
From the expressions contained in Proposition \ref {prop:grosses_expressions} one gets
\begin{equation*}
\frac{1+\frac{4a^2b^2-i2\varrho\beta_1^\theta(\lambda)^2}{A_2^\theta(\lambda)+iB_2^\theta(\lambda)}}{1+\frac{4a^2b^2+i2\varrho\beta_2^\theta(\lambda)^2}{A_2^\theta(\lambda)-iB_2^\theta(\lambda)}}
=\frac{A_2^\theta(\lambda)-iB_2^\theta(\lambda)}{A_2^\theta(\lambda)+iB_2^\theta(\lambda)}\times \frac{A^\theta_2(\lambda)+4a^2b^2+i\big(B^\theta_2(\lambda) -2\varrho\beta_1^\theta(\lambda)^2\big)}{A^\theta_2(\lambda)+4a^2b^2+i\big(2\varrho\beta_2^\theta(\lambda)^2-B^\theta_2(\lambda)\big)}, 
\end{equation*}
and one concludes by observing that the equality
\begin{equation*}
B^\theta_2(\lambda) -2\varrho\beta_1^\theta(\lambda)^2
=2\varrho\beta_2^\theta(\lambda)^2-B^\theta_2(\lambda)
\end{equation*}
always holds.
\end{proof}

For the subsequent computations, it is will be useful too keep in mind the behaviour of the functions
$\beta_j^\theta(\cdot)^2$. For example, $\beta_1^\theta(\cdot)^2$ vanishes at $\lambda_1^\theta\pm2$, takes the value $4\Xi_-(\theta)$ at $\lambda_2^\theta-2$, and the value
$4\Xi_+(\theta)$ at $\lambda_2^\theta+2$.
Similarly, $\beta_2^\theta(\cdot)^2$ vanishes at $\lambda_2^\theta\pm2$,
takes the value  $4\Xi_-(\theta)$ at $\lambda_1^\theta+2$, and the value
$4\Xi_+(\theta)$ at $\lambda_1^\theta-2$.
The function $\beta_1^\theta(\cdot)^2$ is also monotone increasing on the interval
$(\lambda^\theta_1+2, \lambda^\theta_2+2)$, while the function $\beta_2^\theta(\cdot)^2$
is monotone decreasing on the interval $(\lambda^\theta_1-2,\lambda^\theta_2-2)$.

Finally, for the computation of the left-hand side of \eqref{levinsonformularesonances}
let us recall the main statement about eigenvalues proved in 
\cite[Prop.~1.2]{NRT}~:

\begin{Proposition}\label{proposition_kernel}
A value $\lambda\in\R\setminus\T^\theta$ is an eigenvalue of $H^\theta$ if and only if
$$
\KK:=\ker\left(\u+\sum_{\{j\mid\lambda<\lambda_j^\theta-2\}}\frac{\v\;\!\P_j^\theta\v}
{\beta_j^\theta(\lambda)^2}-\sum_{\{j\mid\lambda>\lambda_j^\theta+2\}}
\frac{\v\;\!\P_j^\theta\v}{\beta_j^\theta(\lambda)^2}\right)
\bigcap\left(\cap_{\{j\mid\lambda\in I^\theta_j\}}\ker\big(\P_j^\theta\v\big)\right)
\ne\{0\},
$$
in which case the multiplicity of $\lambda$ equals the dimension of $\KK$.
\end{Proposition}

\subsubsection{The special case \texorpdfstring{$b^2=0$}{b=0}.}

We firstly consider $b^2=0$ and arbitrary $\theta\in (0,\pi)$. In this case, $u_2=1$ and 
$\varrho=a^2$.

\paragraph{Variation on $(\lambda_1^\theta-2,\lambda_2^\theta-2)$.}

From \eqref{A1B1} we obtain
\begin{equation*}
A^\theta_1(\lambda)=2u_1\beta_1^\theta(\lambda)^2\beta_2^\theta(\lambda)^2
+a^2\beta_1^\theta(\lambda)^2 
\quad \hbox{and} \quad 
B^\theta_1(\lambda)=a^2\beta_2^\theta(\lambda)^2.
\end{equation*}
Then the curve $(\lambda_1^\theta-2,\lambda_2^\theta-2)\ni \lambda \mapsto 
A^\theta_1(\lambda)-iB^\theta_1(\lambda) \in \C$ starts at $-4ia^2\Xi_+(\theta)$, ends at $4a^2\Xi_-(\theta)$, and keeps a strictly imaginary negative part on this domain. From this one concludes that 
\begin{equation}\label{eq:b=0parte1}
\Var\left((\lambda_1^\theta-2,\lambda_2^\theta-2)\ni \lambda \mapsto \det S^\theta(\lambda)\right)=-\tfrac12.
\end{equation}

\paragraph{Variation on $(\lambda_2^\theta-2,\lambda_1^\theta+2)$ .}

From \eqref{C2D2} we obtain
\begin{equation*}
A_2^\theta(\lambda)=2u_1\beta_1^\theta(\lambda)^2\beta_2^\theta(\lambda)^2  
\quad \hbox{and} \quad 
B_2^\theta(\lambda)=a^2\big(\beta_1^\theta(\lambda)^2+\beta_2^\theta(\lambda)^2\big).
\end{equation*}
Then the curve  $(\lambda_2^\theta-2,\lambda_1^\theta+2) \ni \lambda \mapsto 
A_2^\theta(\lambda)-iB_2^\theta(\lambda) \in \C$ is a closed curve starting and ending at $-4ia^2\Xi_-(\theta)$. Since $A_2^\theta(\lambda)$ has the sign of $u_1$ and $B_2^\theta(\lambda)>0$, the closed curve lies either in quadrant III for $u_1=-1$,  or in quadrant IV for $u_1=1$. We then infer from Lemma \ref{lem:calculodetslambda} and from \eqref{eq:Sint}
that
\begin{equation}\label{eq:b=0parte2}
\Var \left((\lambda_2^\theta-2,\lambda_1^\theta+2)\ni \lambda \mapsto \det S^\theta(\lambda)\right)=0.
\end{equation}

\paragraph{Variation on $(\lambda_1^\theta+2,\lambda_2^\theta+2)$.}

From \eqref{A3B3} we obtain
\begin{equation*}
A^\theta_3(\lambda)=2u_1\beta_1^\theta(\lambda)^2\beta_2^\theta(\lambda)^2-a^2\beta_2^\theta(\lambda)^2 
\quad \hbox{and} \quad 
B^\theta_3(\lambda)=a^2\beta_1^\theta(\lambda)^2.
\end{equation*}
Then the curve $(\lambda_1^\theta+2,\lambda_2^\theta+2)\ni \lambda \mapsto A^\theta_3(\lambda)-iB^\theta_3(\lambda)\in \C$
starts at $-4a^2\Xi_-(\theta)$, ends at $-4ia^2\Xi_+(\theta)$, and keeps a strictly negative imaginary part on this domain.  From this one infers that 
\begin{equation}\label{eq:b=0parte3}
\Var \left((\lambda_1^\theta+2,\lambda_2^\theta+2)\ni \lambda \mapsto \det
S^\theta(\lambda)\right)=-\tfrac12.
\end{equation}

By summing the three contributions \eqref{eq:b=0parte1}, \eqref{eq:b=0parte2} and 
\eqref{eq:b=0parte3} we get that
$$
 \Var \big(\lambda \to \det S^\theta(\lambda)\big) = -\tfrac12+0-\tfrac12 = -1.
$$
Since the special case $b=0$ does not lead to any resonance, according to Proposition \ref{S(lambda)N=2}, it follows that \eqref{levinsonformularesonances} reads in this case
\begin{equation*}
\# \sigma_{\rm p}(H^\theta)=2-0+(-1)=1.
\end{equation*}
Indeed, by using Proposition \ref{proposition_kernel} one can check that $H^\theta$ has only one eigenvalue above its essential spectrum if $u_1=1$, and below its essential spectrum if $u_1=-1$, 
see also \cite[Example 3.2]{NRT}.

\subsubsection{The special case \texorpdfstring{$a^2=b^2$}{a²=b²}.}

We now turn now our attention to the special case $a^2=b^2>0$ and
arbitrary $\theta \in (0,\pi)$. In order to avoid a 
$1$-periodic system we assume without loss of generality that $u_1=1$ and $u_2=-1$. 
In this case, it follows that $\varrho=0$.

\paragraph{Variation on $(\lambda_1^\theta-2,\lambda_2^\theta-2)$.}

From \eqref{A1B1} we obtain
\begin{equation*}
A^\theta_1(\lambda)=-2\beta_1^\theta(\lambda)^2\beta_2^\theta(\lambda)^2
\quad \hbox{and} \quad 
B^\theta_1(\lambda)=2a^4.
\end{equation*}
Then the curve $(\lambda_1^\theta-2,\lambda_2^\theta-2)\ni \lambda \mapsto 
A^\theta_1(\lambda)-iB^\theta_1(\lambda) \in \C$ 
is a closed curve starting and ending at $-2ia^4$ and lying in quadrant III. 
From this one concludes that 
\begin{equation}\label{eq:a=bparte1}
\Var\left((\lambda_1^\theta-2,\lambda_2^\theta-2)\ni \lambda \mapsto \det S^\theta(\lambda)\right)=0.
\end{equation}

\paragraph{Variation on $(\lambda_2^\theta-2,\lambda_1^\theta+2)$ .}

From \eqref{C2D2} we obtain
\begin{equation*}
A_2^\theta(\lambda)=-2\beta_1^\theta(\lambda)^2\beta_2^\theta(\lambda)^2-2a^4
\quad \hbox{and} \quad 
B_2^\theta(\lambda)=0.
\end{equation*}
In this case one has $\frac{\s_{22}^\theta(\lambda)}{\overline{\s_{11}^\theta(\lambda)}}=1$ for all $\lambda \in (\lambda_2^\theta-2,\lambda_1^\theta+2)$ and hence
\begin{equation}\label{eq:a=bparte2}
\Var \left((\lambda_2^\theta-2,\lambda_1^\theta+2)\ni \lambda \mapsto \det S^\theta(\lambda)\right)=0.
\end{equation}

\paragraph{Variation on $(\lambda_1^\theta+2,\lambda_2^\theta+2)$.}

From \eqref{A3B3} we obtain
\begin{equation*}
A^\theta_3(\lambda)=-2\beta_1^\theta(\lambda)^2\beta_2^\theta(\lambda)^2
\quad \hbox{and} \quad 
B^\theta_3(\lambda)=-2a^4.
\end{equation*}
Then the curve $(\lambda_1^\theta+2,\lambda_2^\theta+2)\ni \lambda \mapsto A^\theta_3(\lambda)-iB^\theta_3(\lambda)\in \C$
is a closed curve starting and ending at $2ia^4$ and lying in quadrant II, and therefore
\begin{equation}\label{eq:a=bparte3}
\Var \left((\lambda_1^\theta+2,\lambda_2^\theta+2)\ni \lambda \mapsto \det
S^\theta(\lambda)\right)=0.
\end{equation}

By summing the three contributions \eqref{eq:a=bparte1}, \eqref{eq:a=bparte2} and 
\eqref{eq:a=bparte3} we get that
$$
 \Var \big(\lambda \to \det S^\theta(\lambda)\big) = 0+0+0 = 0.
$$
Since the special case $a^2=b^2$ does not lead to any resonance, according to Proposition \ref{S(lambda)N=2}, it follows that \eqref{levinsonformularesonances} reads in this case
\begin{equation*}
\# \sigma_{\rm p}(H^\theta)=2-0+0=2.
\end{equation*}
Indeed, by using Proposition \ref{proposition_kernel} one can check that  $H^\theta$ has one eigenvalue above and one eigenvalue below its essential spectrum: the condition for an eigenvalue
at $\lambda \not \in \overline{I^\theta}$ coming from \eqref{eq:eigenvaluebelow} and \eqref{eq:eigenvalueabove} simplifies in 
\begin{equation}\label{eq:vp2}
-\beta_1^\theta(\lambda)^{2}\beta_2^\theta(\lambda)^{2}+a^2b^2=0.
\end{equation}
Since $\beta_1^\theta(\lambda)^{2}\beta_2^\theta(\lambda)^{2}$ vanishes at $\lambda_1^\theta-2$ and at $\lambda_2^\theta+2$, 
$\lim_{\lambda \to \pm \infty}\beta_1^\theta(\lambda)^{2}\beta_2^\theta(\lambda)^{2} =\infty$,
and $\beta_1^\theta(\lambda)^{2}\beta_2^\theta(\lambda)^{2}$ behaves quadratically below and above the essential spectrum,  one infers that \eqref{eq:vp2} has exactly two solutions. 

\subsubsection{A resonant case: \texorpdfstring{$a^2=1$, $b^2=\tfrac12$ and $u_1=u_2=-1$}{a²=1, b²=1/2, and u1=u2=-1}.}

For this example we will fix $a^2=1$, $b^2=\tfrac12$, and $u_1=u_2=-1$, in order to be 
in the white region of Figure \ref{fig:SN2}. As a consequence, the behaviour at thresholds
will depend on $\theta \in (0,\pi)$. We observe that $\varrho=-\frac32$, and fix $\theta_0\in (0,\pi)$
by the condition 
\begin{equation}\label{eq:theta0}
\Xi_+(\theta_0)=\tfrac16.
\end{equation}
Note that this is possible since $\Xi_+$ is a strictly decreasing function on $(0,\pi)$ with 
range $(0,2)$.

\paragraph{Variation on $(\lambda_1^\theta-2,\lambda_2^\theta-2)$.}

From \eqref{A1B1} we obtain
\begin{equation*}
A^\theta_1(\lambda)=\beta_1^\theta(\lambda)^2\big(2\beta_2^\theta(\lambda)^2-\tfrac32\big)
\quad \hbox{and} \quad 
B^\theta_1(\lambda)=1-\tfrac32\beta_2^\theta(\lambda)^2.
\end{equation*}
We first observe that for any $\theta\in (0,\pi)$ the function $B_1^\theta$ is strictly increasing 
on  $(\lambda_1^\theta-2,\lambda_2^\theta-2)$, with 
$B^\theta_1(\lambda_1^\theta-2)=1-6\Xi_+(\theta_0)$ and 
$B_1^\theta(\lambda_2^\theta-2)=1$.
In addition, the following relations clearly hold: 
\begin{equation*}
B_1^\theta(\lambda_1^\theta-2)<0  \text{ for } \theta<\theta_0, \quad B_1^\theta(\lambda_1^\theta-2)=0  \text{ for } \theta=\theta_0, \quad B_1^\theta(\lambda_1^\theta-2)>0 \text{ for } \theta>\theta_0.
\end{equation*}
For $\theta <\theta_0$, we define $\lambda^\theta_B$  by the relation $B_1^\theta(\lambda^\theta_B)=0$.

We can also observe that $A_1^\theta(\lambda^\theta_1-2)=0$ and that the sign of $A_1^\theta(\lambda)$ is given by the sign of $2\beta_2^\theta(\lambda)^2-\tfrac32$. In particular, since $\beta_2^\theta$ is monotone decreasing on $(\lambda_1^\theta-2,\lambda_2^\theta-2)$, 
the function $A_1^\theta$ is negative on this interval if and only if
\begin{equation*}
2\beta_2^\theta(\lambda_1^\theta-2)^2-\tfrac32 \leq 0 \Longleftrightarrow \Xi_+(\theta) \leq \tfrac3{16}. 
\end{equation*}
On the other hand, for $\theta\in (0,\pi)$ verifying the condition $\Xi_+(\theta)> \tfrac3{16}$, the function $A_1^\theta$ on $(\lambda_1^\theta-2,\lambda_2^\theta-2)$ is firstly positive, then
equal to $0$ for some $\lambda^\theta_A$, and continues decreasing until it reaches 
the value $-6\Xi_-(\theta)$ at $\lambda_2^\theta-2$.
Note that since the function $\Xi_+$ is decreasing on $(0,\pi)$ and since $\tfrac3{16}>\tfrac16$, 
it follows that the parameters $\theta$ verifying the previous condition also satisfy 
$\theta<\theta_0$.

Let us now consider the $A_1^\theta-iB_1^\theta$ for $\theta$ satisfying 
$\Xi_+(\theta)> \tfrac3{16}$.
One readily observes that 
\begin{equation*}
\beta_2^\theta(\lambda^\theta_A)^2=\tfrac34 \quad\hbox{ and }\quad \beta_2^\theta(\lambda^\theta_B)^2=\tfrac23, 
\end{equation*} 
from which one deduces that $\lambda_A <\lambda_B$. Hence, the curve
\begin{equation}\label{eq:curve}
(\lambda_1^\theta-2,\lambda_2^\theta-2)\ni \lambda \mapsto A_1^\theta(\lambda)-iB_1^\theta(\lambda)
\end{equation}
remains in quadrants I, II, and III, starting at $(6\Xi_+(\theta)-1)i$ and ending at
$-6\Xi_-(\theta)-i$.
Then, we conclude that the map $\lambda\mapsto \s_{11}^\theta(\lambda)$ goes from $-1$ to $z:=\frac{6\Xi_-(\theta)+i}{6\Xi_-(\theta)-i}$ in the anticlockwise direction. 
Note that $z$ has positive imaginary part and hence $\Arg(z)\in (0,\pi)$. 
Observe also that the same result directly applies to all $\theta<\theta_0$ since the curve
\eqref{eq:curve} does not visit quadrant I if $\Xi_+(\theta)\leq \tfrac3{16}$.

Let us now consider the case $\theta>\theta_0$.
In this case, the curve \eqref{eq:curve} starts at $(6\Xi_+(\theta)-1)i$ in the lower 
part of the imaginary axis, and ends at $-6\Xi_-(\theta)-i$ while 
remaining in quadrant III. In this case, the map $\lambda \mapsto \s_{11}^\theta(\lambda)$ 
goes from $-1$ to $z$ in the clockwise direction. 

For $\theta=\theta_0$, we already know from \eqref{eq:k1} that
$\s_{11}^\theta(\lambda^\theta_1-2)=1$.
We also easily observe that the map 
$$
(\lambda_1^{\theta_0}-2,\lambda_2^{\theta_0}-2)\ni \lambda \mapsto A_1^{\theta_0}(\lambda)-iB_1^{\theta_0}(\lambda)
$$
stays in quadrant III, since $B_1^{\theta_0}(\lambda)>0$ for $\lambda>\lambda_1^{\theta_0}-2$. We then infer that the map
$\lambda \mapsto \s_{11}^{\theta_0}(\lambda)$
goes from $1$ to $z$ while keeping a positive imaginary part,  and hence in the anticlockwise direction.

If we summarise our findings, we have obtained that
\begin{equation}\label{eq:resonantparte1}
\Var\left((\lambda_1^\theta-2,\lambda_2^\theta-2)\ni \lambda \mapsto \det S^\theta(\lambda)\right)=
\begin{cases}
-\frac12-\frac{\Arg(z)}{2\pi} &\text{ for } \theta<\theta_0,\\
-\frac{\Arg(z)}{2\pi} &\text{ for } \theta=\theta_0,\\
\frac12-\frac{\Arg(z)}{2\pi} &\text{ for } \theta>\theta_0.
\end{cases}
\end{equation}

\paragraph{Variation on $(\lambda_2^\theta-2,\lambda_1^\theta+2)$ .}

From \eqref{C2D2} we get that
\begin{equation*}
A_2^\theta(\lambda)=2\beta_1^\theta(\lambda)^2\beta_2^\theta(\lambda)^2-1
\quad \hbox{and} \quad 
B_2^\theta(\lambda)=-\tfrac32(\beta_1^\theta(\lambda)^2+\beta_2^\theta(\lambda)^2).
\end{equation*}
It follows that the map $(\lambda_2^\theta-2,\lambda_1^\theta+2)\ni \lambda \mapsto A_2^\theta(\lambda)-iB_2^\theta(\lambda)$ is a closed curve on $-1+6\Xi_-(\theta)i$ and lying in quadrants I and II. We then conclude that 
\begin{equation}\label{eq:resonantparte2}
\Var \left((\lambda_2^\theta-2,\lambda_1^\theta+2)\ni \lambda \mapsto \det S^\theta(\lambda)\right)=0.
\end{equation}

\paragraph{Variation on $(\lambda_1^\theta+2,\lambda_2^\theta+2)$.}

From \eqref{A3B3} we obtain
\begin{equation*}
A^\theta_3(\lambda)=\beta_2^\theta(\lambda)^2\big(2\beta_1^\theta(\lambda)^2+\tfrac32\big)
\quad \hbox{and} \quad 
B^\theta_3(\lambda)=-1-\tfrac32\beta_1^\theta(\lambda)^2.
\end{equation*}
One readily infers that the map $(\lambda_1^\theta+2,\lambda_2^\theta+2)\ni \lambda \mapsto A^\theta_3(\lambda)-iB^\theta_3(\lambda)$ remains in quadrant I and hence $\s_{22}^\theta(\lambda)$ has a positive imaginary part. Since we also have
\begin{equation*}
\s_{22}^\theta(\lambda_1^\theta+2)=\frac{A^\theta_3(\lambda_1^\theta+2)-iB^\theta_3(\lambda_1^\theta+2)}{A^\theta_3(\lambda_1^\theta+2)+iB^\theta_3(\lambda_1^\theta+2)}=\frac{6\Xi_-(\theta)+i}{6\Xi_-(\theta)-i}=z,
\end{equation*}
and recalling that $\s_{22}^\theta(\lambda_2^\theta+2)=-1$ we deduce that 
\begin{equation}\label{eq:resonantparte3}
\Var \left((\lambda_1^\theta+2,\lambda_2^\theta+2)\ni \lambda \mapsto \det
S^\theta(\lambda)\right)
=-\tfrac12+\tfrac{\Arg(z)}{2\pi}.
\end{equation}
As before, by considering \eqref{levinsonformularesonances} together with \eqref{eq:resonantparte1}, \eqref{eq:resonantparte2}, and \eqref{eq:resonantparte3} and recalling from Proposition \ref{S(lambda)N=2} that for $\theta=\theta_0$ we have a resonance at the bottom of the spectrum we get 
\begin{equation}
\# \sigma_{\rm p}(H^\theta)=\begin{cases}
2-0+\left(-\frac12-\frac{\Arg(z)}{2\pi}+0-\frac12+\frac{\Arg(z)}{2\pi}\right)&=1 \quad\text{ if } \theta<\theta_0,\\
2-\frac12+\left(-\frac{\Arg(z)}{2\pi}+0-\frac12+\frac{\Arg(z)}{2\pi}\right)&=1 \quad\text{ if } \theta=\theta_0,\\
2-0+\left(\frac12-\frac{\Arg(z)}{2\pi}+0-\frac12+\frac{\Arg(z)}{2\pi}\right)&=2 \quad\text{ if } \theta>\theta_0.
\end{cases}
\end{equation}

We illustrate this case with Figure \ref{fig_Vic} which provides the information about 
the number of eigenvalues and their location below the essential spectrum of $H^\theta$.
More precisely, the sign of the determinant of the matrix
\begin{equation*}
\u+ \frac{\v\P_1^\theta \v}{\beta^\theta_1(\lambda)^2}+ \frac{\v\P_2^\theta \v}{\beta^\theta_2(\lambda)^2}.
\end{equation*}
is computed, as a function of $\lambda$ and $\theta$. Depending on its value,
either a red colour or a blue colour is assigned to the point $(\lambda, \theta)$.
Thus, interfaces between a blue region and a red region coincide with
eigenvalues below the essential spectrum, according to Proposition \ref{proposition_kernel}.
In Figure \ref{fig_Vic}, the black curve represents the bottom of the essential spectrum of $H^\theta$.
For most fixed $\theta$, there exists only one interface between the red region and the blue region, meaning that the corresponding operator $H^\theta$ possesses only one eigenvalue.
However, for $\theta$ close to $\pi$, a second interface appears, as emphasised in the magnified part.
Then the corresponding operators $H^\theta$ possess two eigenvalues below the essential spectrum. 
The minimal value $\theta_0$ above which a second eigenvalue appears can be determined by solving the equation \eqref{eq:theta0}.

\begin{figure}
    \centering
    \includegraphics[width=12cm]{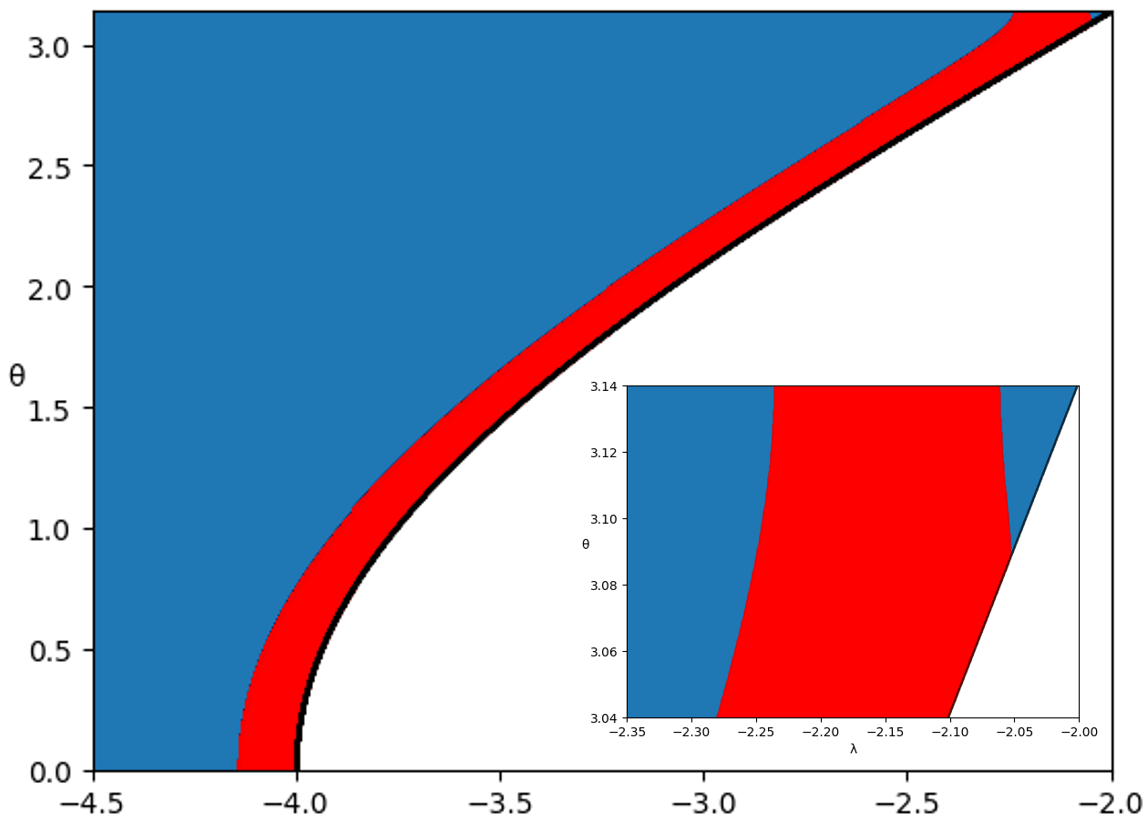}
    \caption{Visual representation of the number of eigenvalues below the essential spectrum of $H^\theta$: For each fixed $\theta$, each interface between a blue and a red region corresponds to an eigenvalue. The magnified picture represents the region close to $\theta=\pi$, where a second interface appears.}
    \label{fig_Vic}
\end{figure}

\section{Appendix}\label{Appendix}

\begin{proof}[Proof of Proposition \ref{prop:step1}]
For $j=1$, the statement corresponds to Lemma \ref{lem:K1}, since 
$$
A(1)\equiv A_1=C_0\big(I;\B(P_1\C^N)\big)=C_0(I)\otimes \B(P_1\C^N).
$$
We then look at the result for $j=2$. In this case the construction leads to the pullback diagram
$$
\xymatrix{ (A_{1}\oplus A_2)\oplus_{\C^2}D_2\ar[r]^{}\ar[d]^{} & A_{1}\oplus A_2\ar[d]\\
D_2\ar[r] & \C^2}
$$
Using the computations of the $K$-theory of  $A_1\oplus A_2$ and $D_2$ 
from Lemmas \ref{lem:K1} and \ref{lem:K2}, the Mayer-Vietoris sequence gives
$$
0\to K_0\big(A(2)\big)\to \Z^2\to \Z^2\to K_1\big(A(2)\big)\to 0.
$$
We will check that the difference homomorphism $K_0(D_2)=\Z^2\to\Z^2= K_0(\C^2)$
is a bijection, and therefore $K_0\big(A(2)\big)=0$ and $K_1\big(A(2)\big)= 0$.
For that purpose, we need the generators of $K_0(D_2)$. 
One choice consists in  the equivalence classes of projections
$$
 \left[t\mapsto\begin{pmatrix} t & \sqrt{t(1-t)}\\\sqrt{t(1-t)}&1-t\end{pmatrix}\right]-\left[t\mapsto \begin{pmatrix}1&0\\0&0\end{pmatrix}\right],\qquad  \left[t\mapsto\begin{pmatrix}0&0\\0&1\end{pmatrix}\right]
$$
for $t\in [0,1]$. By evaluating them at $t=0$ one gets
$$
\left[\begin{pmatrix}0&0\\0&1\end{pmatrix}\right]-\left[\begin{pmatrix}1&0\\0&0\end{pmatrix}\right],\qquad \left[\begin{pmatrix}0&0\\0&1\end{pmatrix}\right]
$$
and therefore we deduce the expected surjection $K_0(D_2)\to K_0(\C^2)$.

By induction, let us now suppose that the result is true for some $j$, and prove it for $j+1$. 
Then we consider the pullback algebra diagram determined by
$$
\xymatrix{ (A(j)\oplus A_{j+1})\oplus_{\B(\C^j)\oplus\C}D_{j+1}\ar[r]^{}\ar[d]^{} & D_{j+1}\ar[d]\\
A(j)\oplus A_{j+1}\ar[r] & \B(\C^j)\oplus\C}
$$
and get the Mayer-Vietoris sequence
$$
0\to K_0\big(A(j+1)\big)\to \Z^{2}\to\Z^{2}\to K_1\big(A(j+1)\big)\to 0.
$$
We again show that the difference homomorphism 
$$
K_0(D_{j+1})=\Z^{2}\to\Z^{2}= K_0\big(\B(\C^j)\oplus \C\big)
$$ 
is a bijection. For that purpose, we need the generators of $K_0(D_{j+1})$. 
Let us set $E_{1,j}\in \C^{j}$ with $E_{1,j}=(1,0,\dots,0)^T$, and let
$E_{j,1}$ denote its transpose. Then we build representative equivalence classes of projections.
These projections are block matrices (the first block being of size $j\times j$) and the choice of representatives are the maps
$$
 \left[t\mapsto\left(\begin{array}{cc}
        t P_1 & \sqrt{t(1-t)}E_{1,j} \\
        \sqrt{t(1-t)}E_{j,1} &1-t
\end{array}\right)\right]
-
\left[t\mapsto\left(\begin{array}{cc}
         P_1 & 0 \\
        0 &0
\end{array}\right)\right]
,\qquad  
\left[t\mapsto\left(\begin{array}{cc}
         0 & 0 \\
        0 & 1
\end{array}\right)\right]
$$
for $t\in [0,1]$.  By evaluating them at $t=0$ one gets
$$
\left[\left(\begin{array}{cc}
        0& 0 \\
        0 &1
\end{array}\right)\right]
-
\left[\left(\begin{array}{cc}
         P_1 & 0 \\
        0 &0
\end{array}\right)\right]
,\qquad
\left[\left(\begin{array}{cc}
         0 & 0 \\
        0 & 1
\end{array}\right)\right]
$$
and therefore we deduce the expected surjection $K_0(D_{j+1})\to K_0\big(\B(C^j)\oplus \C\big)$. This completes the proof.
\end{proof}

\begin{proof}[Proof of Proposition \ref{prop:step2}]
We shall ignore the unit in the following proof, since adding the unit will only add a copy of 
$\Z$ to $K_0(Q_N)$. 
Let us consider the pullback diagram
$$
\xymatrix{Q_N\ar[r]^{}\ar[d]^{} & \big(A(N-1)\oplus A_N\big)\oplus\big(B_N\oplus B(N-1)\big)\ar[d]\\
C_N\ar[r] & \B(\C^{N-1})\oplus\C\oplus \C\oplus \B\big((\C^1)^\bot\big)}
$$
and the corresponding Mayer-Vietoris sequence
$$
0\to K_0(Q_N)\to  \Z^{3}\to\Z^{4}\to K_1(Q_N)\to 0
$$
where the content of Propositions \ref{prop:step1} and \ref{prop:step1bis}
and of Lemma \ref{lem:K3} have been used for the computation of the various $K$-groups.
By examining the generators, we shall show that the difference homomorphism 
$$
K_0(C_N)=\Z^{3}\to\Z^{4}=K_0\Big( \B(\C^{N-1})\oplus\C\oplus \C\oplus \B\big((\C^1)^\bot\big)\Big)
$$ 
is one-to-one, leading to $K_0(Q_N)=0$ and $K_1(Q_N)=\Z$.
In the construction of the generators of $K_0(C_N)$
we use the notation already introduced in the proof of Proposition \ref{prop:step1}.
These generators are equivalence classes of projections.
These projections are made of block matrices (the first block being of size $(N-1)\times (N-1)$) and one choice of representatives are the maps defined by 
\begin{align*}
& \left[t\mapsto \left(\begin{array}{cc}
        t P_1 & \sqrt{t(1-t)}E_{1,N-1} \\
        \sqrt{t(1-t)}E_{N-1,1} &1-t
\end{array}\right)\right]
-
\left[t\mapsto \left(\begin{array}{cc}
         P_1 & 0 \\
        0 &0
\end{array}\right)\right]  \\
&  
\left[t\mapsto \left(\begin{array}{cc}
         0 & 0 \\
        0 & 1
\end{array}\right)\right]
\qquad \hbox{ and }\qquad 
\left[t\mapsto \left(\begin{array}{cc}
         P_1 & 0 \\
        0 & 0
\end{array}\right)\right]
\end{align*}
for $t\in [0,1]$.  By evaluating them at $t=0$ one gets
$$
\left[\left(\begin{array}{cc}
        0& 0 \\
        0 &1
\end{array}\right)\right]
-
\left[\left(\begin{array}{cc}
         P_1 & 0 \\
        0 &0
\end{array}\right)\right]
,\quad
\left[\left(\begin{array}{cc}
         0 & 0 \\
        0 & 1
\end{array}\right)\right]
,\quad
\left[\left(\begin{array}{cc}
         P_1 & 0 \\
        0 & 0
\end{array}\right)\right]
$$
and evaluating at $t=1$ gives (with the first block of size $1\times 1$)
$$
\left[\left(\begin{array}{cc}
         0 & 0 \\
        0 & P_N
\end{array}\right)\right]
,\quad
\left[\left(\begin{array}{cc}
         1 & 0 \\
        0 & 0
\end{array}\right)\right].
$$
The image under the combined evaluation map evidently contains three independent generators of $\Z^4=K_0\Big( \B(\C^{N-1})\oplus\C\oplus \C\oplus \B\big((\C^1)^\bot\big)\Big)$, and that suffices to complete the proof.
\end{proof}


\begin{thebibliography}{1}

\bibitem{AR23}
A.~Alexander, A.~Rennie,
\emph{Levinson's theorem as an index pairing},
J. Funct. Anal.  28 no.~5, paper No. 110287, 2024.

\bibitem{BSB}
J.~Bellissard, H.~Schulz-Baldes,
\emph{Scattering theory for lattice operators in dimension $d\geq3$}, 
Rev. Math. Phys. 24 no.~8, 1250020, 51 pp., 2012.

\bibitem{Cha00}
A.~Chahrour,
\emph{On the spectrum of the Schr\"odinger operator with periodic surface potential},
Lett. Math. Phys. 52 no.~3, 197--209, 2000.

\bibitem{CS00}
A.~Chahrour, J.~Sahbani,
\emph{On the spectral and scattering theory of the Schr\"odinger operator with
surface potential},
Rev. Math. Phys. 12 no.~4, 561--573, 2000.

\bibitem{CoM}  
A.~Connes, H.~Moscovici,  
\emph{ Type III and spectral triples}, 
Aspects Math. E38, Friedr. Vieweg \& Sohn, Wiesbaden, 57--71, 2008.

\bibitem{Fra03}
R.~Frank,
\emph{On the scattering theory of the {L}aplacian with a periodic boundary
condition. I. {E}xistence of wave operators},
Doc. Math. 8, 547--565, 2003.

\bibitem{Fra04}
R.~Frank, R.~Shterenberg,
\emph{On the scattering theory of the Laplacian with a periodic boundary
condition. II. Additional channels of scattering},
Doc. Math. 9, 57--77, 2004.

\bibitem{HR}
N.~Higson, J.~Roe, 
\emph{Analytic K-homology},
Oxford Mathematical Monographs, Oxford University Press, Oxford, 2000.  

\bibitem{IR1}
H.~Inoue, S.~Richard,
\emph{Index theorems for Fredholm, semi-Fredholm, and almost periodic operators: all in one example},
J. Noncommut. Geom. 13 no.~4, 1359--1380, 2019.

\bibitem{IR2}
H.~Inoue, S.~Richard,
\emph{Topological Levinson's theorem for inverse square potentials: complex, infinite, but not exceptional},
Rev. Roumaine Math. Pures Appl. 64 no.~2-3, 225--250, 2019.

\bibitem{IT19}
H.~Inoue, N.~Tsuzu,
\emph{Schr\"odinger Wave Operators on the Discrete Half-Line},
Integr. Equ. Oper. Theory 91 no.~42, 2019.

\bibitem{JL00}
V.~Jak$\check {\mathrm s }$i\' c, Y.~Last,
\emph{Corrugated surfaces and a.c. spectrum},
Rev. Math. Phys. 12 no.~11, 1465--1503, 2000.

\bibitem{Kat95}
T.~Kato,
\emph{Perturbation theory for linear operators},
Classics in Mathematics. Springer-Verlag, Berlin, 1995.

\bibitem{KR06}
J.~Kellendonk, S.~Richard,
\emph{Levinson's theorem for Schr\"odinger operators with point interaction: a topological approach},
Phys. A 39 no.~46, 14397--14403, 2006.

\bibitem{KR07}
J.~Kellendonk, S.~Richard,
\emph{Topological boundary maps in physics},
in Perspectives in operator algebras and mathematical physics, pages 105--121, Theta Ser. Adv. Math. 8, Theta, Bucharest, 2008.

\bibitem{KR08}
J.~Kellendonk, S.~Richard,
\emph{The topological meaning of Levinson's theorem, half-bound states included},
J. Phys. A 41 no.~29, 295207, 2008.

\bibitem{KR08wind} J.~Kellendonk, S.~Richard,
\emph{On the structure of the wave operators in one-dimensional potential scattering}, 
Mathematical Physics Electronic Journal 14, 1--21, 2008.

\bibitem{Lev}
N.~Levinson, 
\emph{On the uniqueness of the potential in a Schr\"odinger equation for a given
asymptotic phase}, Danske Vid. Selsk. Mat.-Fys. Medd. 25 no.~9, 29 pp., 1949.

\bibitem{NRT}
H.S.~Nguyen, S.~Richard, R.~Tiedra de Aldecoa, 
\emph{Discrete Laplacian in a half‐space with a periodic surface potential I: Resolvent expansions, scattering matrix, and wave operators}, Math. Nachr.  295 no.~5, 912--949, 2022. 

\bibitem{NPR}
F.~Nicoleau, D.~Parra, S.~Richard,
\emph{Does Levinson's theorem count complex eigenvalues?}
J. Math. Phys. 58, 102101-1 - 102101-7, 2017. 

\bibitem{PR}
K.~Pankrashkin, S.~Richard,
\emph{One-dimensional Dirac operators with zero-range interactions: spectral, scattering, and topological results},
J. Math. Phys. 55 no.~6, 062305, 2014.


\bibitem{RiRIMS}
S.~Richard, 
\emph{Discrete Laplacian in a half-space with a periodic surface potential: an overview of analytical investigations},
RIMS Kokyuroku 2195, 18 - 26, Research Institute for Mathematical Sciences, Kyoto University, Kyoto, 2021. 

\bibitem{Ri}
S.~Richard,
\emph{Levinson's theorem: an index theorem in scattering theory},
in Spectral Theory and Mathematical Physics, volume 254 of Operator Theory: Advances and Applications, pages 149--203, Birkh\"auser, Basel, 2016.

\bibitem{RT10}
S.~Richard, R.~Tiedra~de Aldecoa,
\emph{New formulae for the wave operators for a rank one interaction},
Integral Equations and Operator Theory 66, 283--292, 2010. 

\bibitem{RT16}
S.~Richard, R.~Tiedra~de Aldecoa,
\emph{Resolvent expansions and continuity of the scattering matrix at embedded
thresholds: the case of quantum waveguides},
Bull. Soc. Math. France 144 no.~2, 251--277, 2016.

\bibitem{RT17}
S.~Richard, R.~Tiedra~de Aldecoa,
\emph{Spectral and scattering properties at thresholds for the {L}aplacian in a
half-space with a periodic boundary condition},
J. Math. Anal. Appl. 446 no.~2, 1695--1722, 2017.

\bibitem{SB16}
H.~Schulz-Baldes,
\emph{The density of surface states as the total time delay},
Lett. Math. Phys. 106 no.~4, 485--507, 2016.

\end{thebibliography}
\end{document}